\documentclass[sigconf,nonacm]{acmart}

\usepackage{multicol}
\usepackage{amsmath}

\usepackage{pgf}
\usepackage{tikz}
\usetikzlibrary{arrows,automata}
\usetikzlibrary{positioning}
\usetikzlibrary{shapes.geometric}
\usetikzlibrary{plotmarks}
\usetikzlibrary{patterns}
\usetikzlibrary{calc}
\usetikzlibrary{fit}
\usetikzlibrary{shapes.misc}
\usetikzlibrary{backgrounds}

\usetikzlibrary{pgfplots.groupplots}
\usetikzlibrary{shapes.multipart}

\usepackage{amsthm}
\usepackage{thmtools}
\usepackage{thm-restate}

\makeatletter
\newtheoremstyle{sig}
  {}
  {}
  {\itshape}
  {}
  {\scshape}
  {.}
  {.5em}
  {#1\@ifnotempty{#2}{ #2}\thmnote{\quad(#3)}}%
\makeatother

\theoremstyle{sig}

\usepackage{mathpartir}
\usepackage{paralist}
\usepackage{listings}
\usepackage{color}
\usepackage{graphicx}
\usepackage{float}
\usepackage{textcomp}
\usepackage{xspace}
\usepackage{pifont}
\usepackage{multirow}
\usepackage{algorithm}
\usepackage[noend]{algpseudocode}
\usepackage{pgfplots,booktabs}
\usepackage{enumitem}
\usetikzlibrary{snakes}
\tikzstyle{printersafe}=[snake=snake,segment amplitude=0 pt]
\usepackage{dirtree}
\usepackage{listings}
\usepackage{syntax}
\usepackage{bm}
\usepackage{threeparttable}
\usepackage{mathtools}
\usepackage{diagbox}
\usepackage{tabularx}
\usepackage{xifthen}
\usepackage{cleveref}
\usepackage{color}
\usepackage{xcolor}
\usepackage{threeparttable}
\usepackage{bbding}

\usepackage{fancyhdr}
\usepackage[normalem]{ulem}

\usepackage{enumitem}
\setlistdepth{15} 
\newlist{longitemize}{itemize}{15}
\setlist[longitemize,1]{label=\textbullet}
\setlist[longitemize,2]{label=\textbullet}
\setlist[longitemize,3]{label=\textbullet}
\setlist[longitemize,4]{label=\textbullet}
\setlist[longitemize,5]{label=\textbullet}
\setlist[longitemize,6]{label=\textbullet}
\setlist[longitemize,7]{label=\textbullet}
\setlist[longitemize,8]{label=\textbullet}
\setlist[longitemize,9]{label=\textbullet}
\setlist[longitemize,10]{label=\textbullet}
\setlist[longitemize,11]{label=\textbullet}
\setlist[longitemize,12]{label=\textbullet}
\setlist[longitemize,13]{label=\textbullet}
\setlist[longitemize,14]{label=\textbullet}
\setlist[longitemize,15]{label=\textbullet}

\usepackage{rotating}
\usepackage{multirow}

\usepackage{subcaption}

\usepackage{adjustbox}
\usepackage{array}

\newcommand{\techReportAppendix}[1]{%
	\ifthenelse{\equal{\shortVersion}{true}}%
		{\cite{techReport}}%
        {Appendix~\ref{#1}}%
}%

\newcommand{\techReportAppendices}[2]{%
	\ifthenelse{\equal{\shortVersion}{true}}%
		{\cite{techReport}}%
		{Appendices~\ref{#1}--\ref{#2}}%
}%

\newcommand{\onlyTechReport}[1]{%
	\ifthenelse{\equal{\shortVersion}{true}}%
		{}%
		{#1}%
}%

\newcommand{\onlyShortVersion}[1]{%
	\ifthenelse{\equal{\shortVersion}{true}}%
		{#1}%
		{}%
}%

\newcommand{\mypara}[1]{\smallskip\noindent\emph{\textbf{{#1:}}}}

\newcommand{\ie}{\emph{i.e.},\xspace}
\newcommand{\eg}{\emph{e.g.},\xspace}

\definecolor{ao(english)}{rgb}{0.0, 0.5, 0.0}
\definecolor{royalblue(web)}{rgb}{0.25, 0.41, 0.88}

\newcommand{\setCommentColor}[1]{%
	\ifthenelse{\equal{#1}{bk}}%
		{\colorlet{colorVar}{red!50}}%
		{\ifthenelse{\equal{#1}{kg}}%
			{\colorlet{colorVar}{blue}}%
			{\ifthenelse{\equal{#1}{mg}}%
				{\colorlet{colorVar}{ao(english)}}%
			{\ifthenelse{\equal{#1}{jr}}%
				{\colorlet{colorVar}{magenta}}%
				{\ifthenelse{\equal{#1}{zw}}%
					{\colorlet{colorVar}{orange}}%
					{\ifthenelse{\equal{#1}{gm}}%
						{\colorlet{colorVar}{cyan}}
						{}
					}%
				}%
			}%
		}%
	}%
}

\newcommand{\commentAuthor}[1]{%
	\ifthenelse{\equal{#1}{bk}}%
		{Boris:\ }%
		{\ifthenelse{\equal{#1}{kg}}%
			{Klaus:\ }%
			{\ifthenelse{\equal{#1}{mg}}%
				{Marco:\ }%
			{\ifthenelse{\equal{#1}{jr}}%
				{Jan:\ }%
				{\ifthenelse{\equal{#1}{zw}}%
					{Zilong:\ }%
					{\ifthenelse{\equal{#1}{gm}}%
						{Gideon:\ }%
						{}
					}%
				}%
			}%
		}%
	}%
}

\definecolor{codegray}{gray}{0.95}

\theoremstyle{definition}
\newtheorem{definition}{Definition}
\newtheorem{example}{Example}
\newtheorem{contract}{Contract}
\theoremstyle{theorem}

\newcommand{\powerset}[1]{2^{#1}}

\newcommand{\itype}{\textit{ITYPE}}
\newcommand{\liT}{\texttt{li}}
\newcommand{\divT}{\texttt{div}}  
\newcommand{\imm}{\texttt{imm}}

\newcommand{\atoms}{\textit{Atoms}}
\newcommand{\myvalue}[1]{\llbracket#1\rrbracket}

\newcommand{\archStates}{\textsc{ArchState}}
\newcommand{\uarchStates}{\mu\textsc{ArchState}}
\newcommand{\archStep}{\textsc{Isa}}
\newcommand{\archEval}{\archStep^*}
\newcommand{\implStates}{\textsc{ImplState}}
\newcommand{\initImplStates}{\textsc{InitState}}
\newcommand{\implStep}{\textsc{Impl}}
\newcommand{\implEval}{\implStep^*}
\newcommand{\implEvalFilter}[1]{\implStep^*| #1}
\newcommand{\archProj}[1]{{#1}_{\textsc{Arch}}}
\newcommand{\uarchProj}[1]{{#1}_{\textsc{µArch}}}
\newcommand{\ctrObs}{\textsc{CtrObs}}
\newcommand{\atkObs}{\textsc{AtkObs}}
\newcommand{\ctrsat}[2]{\arch, \uarch \vdash \textsc{CtrSat}(#1,#2)}
\newcommand{\uarchctrsat}[3]{\uarch \vdash_{#3} \textsc{CtrSat}(#1,#2)}

\newcommand{\Nat}{\mathbb{N}}

\newcommand{\kywd}[1]{\mathbf{#1}}

\definecolor{Blue3}{HTML}{0000CD}
\newcommand{\obsKywd}[1]{\textcolor{Blue3}{\mathtt{#1}}}

\newcommand{\startObsKywd}[1]{\obsKywd{start}}
\newcommand{\commitObsKywd}[1]{\obsKywd{commit}} 
\newcommand{\rollbackObsKywd}[1]{\obsKywd{rollback}}

\newcommand{\commitObs}[1]{\commitObsKywd{}} %
\newcommand{\rollbackObs}[1]{\rollbackObsKywd{}} %

\newcommand{\ctr}{\textsc{Ctr}}

\newcommand{\fetch}[1]{
\ifthenelse{\equal{#1}{}}{\kywd{fetch}}{\kywd{fetch}}%
}

\definecolor{Blue3}{HTML}{0000CD}
\definecolor{Green4}{HTML}{008B00}
\definecolor{Red3}{HTML}{CD0000}
\definecolor{orange}{rgb}{0.8, 0.47, 0.196}
\lstdefinestyle{Cstyle}
{
	frame = tb,
  belowskip=.4\baselineskip,
  aboveskip=.4\baselineskip,
  	showstringspaces = false,
  	breaklines = true,
  	breakatwhitespace = true,
  	tabsize = 3,
  	numbers = left,
    stepnumber = 1,
    numberstyle = \tiny\color{gray},
    language = {[ANSI]C},
    alsoletter={.\$},
    basicstyle={\ttfamily\color{black}},
    keywordstyle={\ttfamily\color{Blue3}},
    keywordstyle=[2]{\ttfamily\color{Green4}},
    keywordstyle=[3]{\ttfamily\color{orange}},
    keywordstyle=[4]{\ttfamily\color{violet}},
    otherkeywords = {skip,not},
    morekeywords = [2]{A,B},
    morekeywords = [3]{},
	morekeywords = [4]{y,x,z,w, size,size_A,k,temp},
	morecomment=[l][\small\itshape\color{purple!40!black}]{//},
	sensitive=true,
}

\usepackage{letltxmacro}
\newcommand*{\SavedLstInline}{}
\LetLtxMacro\SavedLstInline\lstinline
\DeclareRobustCommand*{\lstinline}{%
  \ifmmode
    \let\SavedBGroup\bgroup
    \def\bgroup{%
      \let\bgroup\SavedBGroup
      \hbox\bgroup
    }%
  \fi
  \SavedLstInline
}

\makeatletter
\newcommand*\@lbracket{[}
\newcommand*\@rbracket{]}
\newcommand*\@colon{:}
\newcommand*\colorIndex{%
    \edef\@temp{\the\lst@token}%
    \ifx\@temp\@lbracket \color{black}%
    \else\ifx\@temp\@rbracket \color{black}%
    \else\ifx\@temp\@colon \color{black}%
    \else \color{orange}%
    \fi\fi\fi
}
\makeatother

\lstdefinestyle{verilog-style}
{
    language=Verilog,
    basicstyle=\footnotesize\ttfamily,
    keywordstyle=\color{blue},
    identifierstyle=\color{black},
    commentstyle=\color{cadmiumgreen},
    numbers=left,
	morekeywords = [2]{leak, on, monitor},
	keywordstyle = [2]\color{punct}\ttfamily,
    numberstyle=\tiny\color{black},
    extendedchars=true,
    numbersep=10pt,
    tabsize=4,
	frame=none,
    moredelim=*[s][\colorIndex]{[}{]},
    literate=*{:}{:}1,
    xleftmargin=5.0ex,
    captionpos=b,
    escapechar=\$,
    escapeinside={(*}{*)}
}
\colorlet{punct}{red!60!black}
\definecolor{delim}{RGB}{20,105,176}
\colorlet{numb}{magenta!60!black}
\definecolor{lightgreen}{HTML}{669900}		%
\definecolor{bluegreen}{HTML}{33997e}		%
\definecolor{brightube}{rgb}{0.82, 0.62, 0.91}
\definecolor{aquamarine}{rgb}{0.5, 1.0, 0.83}
\definecolor{cadmiumgreen}{rgb}{0.0, 0.42, 0.24}

\newcolumntype{R}[2]{%
    >{\adjustbox{angle=#1,lap=\width-(#2)}\bgroup}%
    l%
    <{\egroup}%
}
\newcommand*\rot{\multicolumn{1}{R{45}{1em}}}

\usetikzlibrary{shapes.geometric, arrows, positioning}

\newcommand{\tool}{\textsc{LeaSyn}}

\newcommand{\arch}{\archStep} %
\newcommand{\uarch}{\implStep} %

\newcommand{\isasat}[3]{#1 \equiv_{#3} #2}
\renewcommand{\isasat}[3]{#1 \vdash_{#3} #2}

\newcommand{\atk}{\textsc{Atk}}
\renewcommand{\contract}{\textsl{\textsc{LC}}}

\newcommand{\template}{\mathbb{T}}

\newcommand{\revision}[1]{{#1}}

\crefname{section}{\S}{\S\S}
\crefname{subsection}{\S}{\S\S}

\title{Synthesis of Sound and Precise Leakage Contracts for Open-Source RISC-V Processors}

\author{Zilong Wang}
\affiliation{%
  \institution{IMDEA Software Institute}
  \city{Madrid}
  \country{Spain}
}
\additionalaffiliation{Universidad Polit\'ecnica de Madrid}
\email{zilong.wang@imdea.org}

\author{Gideon Mohr}
\affiliation{%
  \institution{Saarland University}
  \city{Saarbr\"ucken}
  \country{Germany}
}
\email{s8gimohr@stud.uni-saarland.de}

\author{Klaus von Gleissenthall}
\affiliation{%
  \institution{Vrije Universiteit Amsterdam}
  \city{Amsterdam}
  \country{Netherlands}
}
\email{k.freiherrvongleissenthal@vu.nl}
  
\author{Jan Reineke}
\affiliation{%
\institution{Saarland University}
\city{Saarbr\"ucken}
\country{Germany}
}
\email{reineke@cs.uni-saarland.de}

\author{Marco Guarnieri}
\affiliation{%
\institution{IMDEA Software Institute}
\city{Madrid}
\country{Spain}
}
\email{marco.guarnieri@imdea.org}


\begin{document}
\begin{abstract}
    Leakage contracts have been proposed as a new security abstraction at the instruction set architecture level. 
    Leakage contracts aim to capture the information that processors may leak via microarchitectural side channels.
    Recently, the first tools have emerged to verify whether a processor satisfies a given contract.
    However, coming up with a contract that is both sound and precise for a given processor is challenging, time-consuming, and error-prone, as it requires in-depth knowledge of the timing side channels introduced by microarchitectural optimizations. %

    In this paper, we address this challenge by proposing \tool{}, the first tool for automatically synthesizing leakage contracts that are both \emph{sound} and \emph{precise} for processor designs at register-transfer level.
    Starting from a user-provided contract template that captures the space of possible contracts, %
    \tool{} automatically constructs a contract, alternating between 
        contract synthesis, which ensures precision based on an empirical characterization of the processor's leaks, 
    and contract verification, which ensures soundness.\looseness=-2 

    Using \tool{}, we automatically synthesize contracts for six open-source RISC-V CPUs for a variety of contract templates.
    Our experiments indicate that \tool{}'s contracts are sound and more precise (\ie represent the actual leaks in the target processor more faithfully) than contracts constructed by existing approaches.\looseness=-1
\end{abstract}


\Crefformat{section}{\S#2#1#3}
\crefformat{section}{\S#2#1#3}
\Crefformat{subsection}{\S#2#1#3}
\crefformat{subsection}{\S#2#1#3}
\Crefformat{subsubsection}{\S#2#1#3}
\crefformat{subsubsection}{\S#2#1#3}

\crefrangeformat{section}{\S#3#1#4--#5#2#6}
\Crefrangeformat{section}{\S#3#1#4--#5#2#6}
\crefrangeformat{subsection}{\S#3#1#4--#5#2#6}
\Crefrangeformat{subsection}{\S#3#1#4--#5#2#6}
\crefrangeformat{subsubsection}{\S#3#1#4--#5#2#6}
\Crefrangeformat{subsubsection}{\S#3#1#4--#5#2#6}

\newcommand{\shortVersion}{false}

\maketitle

\begin{figure*}[ht] 
    ~\hfill\includegraphics[width=\textwidth]{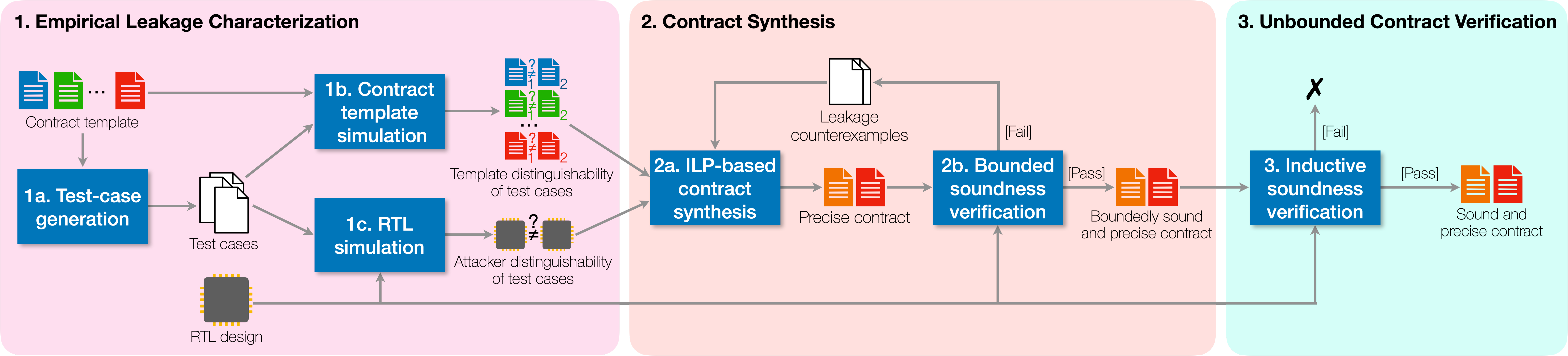}
    \vspace{-17pt}
    \caption{Overview of \tool{} synthesis approach}
    \label{fig:approach-overview}
    \vspace{-11pt}
  \end{figure*}
\section{Introduction}\label{sec:introduction}

Microarchitectural attacks~\cite{Aciicmez07,RistenpartTSS09,Yarom14,Liu15,Yarom2017,spectre2019,ret2spec2018,stecklina2018lazyfp,Aldaya2019,fallout2019,zombieload2019,ridl2019,Weber0NSR21,GerlachWZS23,PurnalBPV23} rely on subtle but measurable hardware side effects on a computation's execution time to compromise otherwise secure programs.
The exploited side effects  result from optimizations---like caching and speculative execution---implemented in a CPU's microarchitecture.
Defending against such attacks is challenging since they are ``invisible'' at the level of the instruction set architecture (ISA), the traditional hardware-software interface familiar to programmers.

\emph{Leakage contracts}~\cite{contracts2021} have recently been proposed as a new security abstraction at ISA-level to fill this gap and to serve as a rigorous foundation for writing programs that are resistant against microarchitectural attacks.
In a nutshell, a leakage contract specifies at ISA-level what information a program may leak via microarchitectural side channels.
Leakage contracts have been successfully applied to reason about the security of hardware~\cite{Wang23,oleksenko2023hide,10.1145/3676641.3716247} and software~\cite{future-mu-archs,spectector2020,fabian202automatic,10.1145/3704867,patrignani2021exorcising} against microarchitectural leaks.\looseness=-1

To prevent microarchitectural attacks, programmers can use leakage contracts as a guideline for secure programming by ensuring that secret information does not influence any of the leakage sources identified in the contracts.
To be useful for this, however, leakage contracts should be \emph{sound}, \ie they should capture all leaks in the underlying processor, and \emph{precise}, \ie they should expose as part of the contract only the information that is actually leaked.
Unsound contracts are problematic since they could lead to insecure code.
At the same time, imprecise contracts may unnecessarily rule out programs as ``insecure'' and introduce performance overheads, as they lead to overly defensive code.

Deriving a sound and precise leakage contract for a given processor is challenging.
Manually constructing a leakage contract for a modern processor is a time-consuming~\cite{speculation-at-fault} and error-prone endeavour: 
it requires a deep understanding of the processor's microarchitecture and the leakage introduced by its optimizations.
Automatic synthesis of leakage contracts promises to address the
impracticalities of manually constructing a contract~\cite{Mohr24,conjunct,h-houdini}. %
However, existing approaches for synthesis are limited:
The approach by Mohr et al.~\cite{Mohr24} can synthesize a precise contract directly from a register-transfer level (RTL) processor design, but it cannot guarantee the contract's soundness, \ie the contract may miss leaks, compromising its usefulness for secure programming.
The approach by Dinesh et al.~\cite{conjunct,h-houdini}, on the other hand, synthesizes sound but (very) imprecise contracts, \eg the synthesized contracts entirely rule out the use of any branch or memory instructions.
Wang et al.~\cite{Wang23} have proposed \textsc{LeaVe}, a tool that can verify the soundness of  leakage contracts against RTL designs, and they have used \textsc{LeaVe} to synthesize sound contracts with a human in the loop.
As we show in \Cref{sec:evaluation}, the resulting contracts are still fairly imprecise.

In this paper, we propose a methodology to automatically synthesize leakage contracts that are  \emph{sound} and \emph{precise} given an RTL processor design.
Next, we describe our main contributions.

\mypara{Synthesis methodology}
We propose a methodology (\Cref{sec:synthesis}) to synthesize sound and precise contracts given an RTL processor design. %
Our methodology starts from (a) the RTL design of the target processor, and (b) a user-provided \emph{contract template} that consists of a set of \emph{contract atoms}, each capturing a potential instruction-level leak. 
Crucially, the template defines the search space for contract synthesis and any subset of the template's atoms is a candidate contract.
Figure~\ref{fig:approach-overview} illustrates our approach, which consists of three phases:\looseness=-1
\begin{compactenum}[(1)]
    \item \textit{Empirical Leakage Characterization:} 
    We start by characterizing a processor's leakage based on a set of test cases (\Cref{sec:synthesis:testing}), where a test case consists of a pair of programs with their corresponding inputs.
    These test cases are automatically generated based on the given contract template, and they aim at exercising the processor to identify leaks.
    The test cases are classified in terms of {attacker distinguishability}, \ie whether the attacker can distinguish the pair of executions or not, and template distinguishability, \ie which sets of atoms from the template distinguish each test case.

    \item \textit{Contract Synthesis:} 
    We build a candidate contract for the target processor by alternating between a synthesis phase (\Cref{sec:synthesis:synthesis}) and a verification phase (\Cref{sec:synthesis:bounded-verification}).
    We start by synthesizing  the most precise candidate contract that captures all leaks exposed in the test cases, based on the attacker and template distinguishability relations.
    To ensure that the candidate contract is sound, \ie it captures all leaks, even those not exercised by the initial test cases, we perform a bounded verification step.
    Whenever we find a leakage counterexample, \ie an attacker-distinguishable test case that is contract indistinguishable under the current candidate contract, this counterexample is added to the set of test cases, and a new contract candidate is synthesized.

    \item \textit{Unbounded Contract Verification:}
    When we cannot find further counterexamples with bounded verification, we perform a final {unbounded verification} step to ensure that indeed all possible leaks in the target processor are captured by the candidate contract (\Cref{sec:synthesis:unbounded-verification}).
    If this passes, the contract is provably sound for the target processor, and we return it to the user.
\end{compactenum}
In contrast to prior approaches for contract synthesis~\cite{Mohr24,conjunct,h-houdini}, our methodology ensures that the synthesized contract is both \emph{sound}, \ie it captures all leaks, and \emph{precise}, \ie it distinguishes as few attacker-indistinguishable test cases as possible.

\mypara{\tool{} synthesis tool}
We implement our methodology in \tool{} (\Cref{sec:implementation}), a tool for synthesizing leakage contracts from processor designs in Verilog.
\tool{} implements all the steps of our methodology: it uses integer-linear programming for contract synthesis, a bounded model checker for bounded soundness verification, and the \textsc{LeaVe} contract verification tool~\cite{Wang23} for unbounded soundness verification.
\tool{} also implements a family of different contract templates capturing instruction-level leaks, which are implemented of top of the RISC-V Formal Interface~\cite{rvfi} to simplify applying \tool{} to new targets.

\mypara{Evaluation}
We validate our methodology by synthesizing  leakage contracts for six open-source RISC-V processors from three different core families (\Cref{sec:evaluation}).
We do so by (a) defining a variety of contract templates capturing different classes of instruction-level leaks (e.g., leaks through control flow, variable-latency instructions, and memory accesses), and (b) using \tool{} to synthesize the most precise and sound contracts satisfied by the studied target processors, against an attacker that observes when instructions retire.
Our experiments confirm that \tool{} can successfully synthesize sound and precise leakage contracts for all of our targets in less than 48 hours. %
Furthermore, the contracts synthesized by \tool{} are significantly more precise than existing  contracts derived either manually~\cite{Wang23} or automatically~\cite{conjunct,h-houdini,rtl2mupath,Mohr24}.

\mypara{Summary}
To summarize, we make the following contributions:
\begin{compactitem}
    \item We propose a methodology for synthesizing sound and precise leakage contracts for RISC-V processors at RTL. 
    \item We implement our methodology in a tool called \tool{}.
    \item We evaluate \tool{} on six open-source RISC-V processors, demonstrating that it can efficiently synthesize sound and precise  contracts.
    \item We show that the contracts synthesized by \tool{} are significantly more precise than those derived in prior work.
\end{compactitem}

\section{Overview}\label{sec:overview}

We now present the core aspects of our approach with an example. 
\subsection{A simple processor}\label{sec:overview:processor}

We start by introducing a simple ISA and a CPU implementing
it.\looseness=-1

\mypara{Instruction Set}
We consider a simple instruction set $\arch$ that supports loading immediate
values into registers and performing (integer) divisions between them.
The architectural state consists of $n$ 32-bit registers, representing
signed integers, which we  identify using $\texttt{R1}, \ldots, \texttt{Rn}$.
It also contains a dedicated register $\mathtt{PC}$ holding the program counter. 
The instruction set consists of two instructions: 
\begin{asparaenum}
  \item The \textbf{load immediate} instruction \liT\ \texttt{RD}, \texttt{imm} sets the
  value of the destination register $\mathtt{RD}$ to $\mathtt{imm}$.
  $\mathtt{RD}$ is a register identifier in $\mathtt{R1}, \ldots, \mathtt{Rn}$
  and $\mathtt{imm}$ is a 32-bit value.
  \item The \textbf{division} instruction \divT\ \texttt{RD}, \texttt{RS1}, \texttt{RS2} divides $\mathtt{RS1}$ by $\mathtt{RS2}$ and stores the result in $\mathtt{RD}$.
  $\mathtt{RD}, \mathtt{RS1}, \mathtt{RS2}$ are register identifiers in
  $\mathtt{R1}, \ldots, \mathtt{Rn}$.
  Inspired by the RISC-V ISA~\cite{Waterman:EECS-2014-54}, upon a division by zero, \ie when the value of $\mathtt{RS2}$
  is $0$, the value of $\mathtt{RD}$ is set to $-1$.
\end{asparaenum}

\mypara{Processor}
We consider a simple two-stage pipelined implementation $\uarch$ of $\arch$.
The \emph{fetch/decode stage} takes one cycle.
It loads the current instruction at the address in $\mathtt{PC}$, decodes it into instruction type and operands, and increments the program counter.

The \emph{execute stage} works as follows:
\begin{asparaitem}
  \item Load immediate instructions \liT\ \texttt{RD}, \texttt{imm} are executed in a single
  cycle by writing the  value $\mathtt{imm}$ to the destination register
  $\mathtt{RD}$.
  \item Division instructions \divT\ \texttt{RD}, \texttt{RS1}, \texttt{RS2} are handled by a dedicated division unit. 
  The unit takes 32 cycles to execute all divisions except when  $\mathtt{RS2}$ is $0$ or $1$, in which case it takes 1 cycle. 
\end{asparaitem}

Whenever the execute stage takes more than one cycle, the fetch/decode stage gets
stalled.

\subsection{Modeling leaks with leakage contracts}\label{sec:overview:contracts}

Next, we show how leakage contracts can be used to capture---at ISA level---the leaks existing in the CPU  $\uarch$ from~\Cref{sec:overview:processor}. 

\mypara{Leakage}
Executing instructions on $\uarch$ can leak information about the value of some of the registers.
For instance, consider an attacker $\atk$ that  observes at which cycles $\uarch$ retires each instruction.
The attacker can distinguish between the following pairs of program executions, where we represent each execution by a pair $\langle p \mid \sigma \rangle$ 
consisting of a program $p$ and an initial state $\sigma$:
\begin{compactenum}
  \item $\langle \mathtt{div}\ \texttt{R1}, \texttt{R2}, \texttt{R3} \mid \sigma_1 \rangle$ \emph{vs.} $\langle \mathtt{div}\ \texttt{R1}, \texttt{R2}, \texttt{R3} \mid \sigma_2 \rangle$, for $\sigma_1, \sigma_2$, s.t. $\sigma_1(\texttt{R3}) = 2$ and $\sigma_2(\texttt{R3}) = 1$. 
  The two executions are attacker distinguishable since executing $\mathtt{div}$ takes 33 cycles (1 cycle for fetch/decode and 32 for the execution stage) in the first case, but only 2 cycles in the second case.
  \item $\langle \mathtt{div}\ \texttt{R1}, \texttt{R2}, \texttt{R3} \mid \sigma_1 \rangle$ \emph{vs.} $\langle \mathtt{li}\ \texttt{R1}, \mathtt{0x0} \mid \sigma_3 \rangle$ for some $\sigma_3$, are also attacker distinguishable since $\mathtt{div}$ takes 33 cycles starting from $s_1$ whereas $\mathtt{li}$ only takes 2 cycles.
\end{compactenum}
The attacker, however, cannot distinguish the following executions:
\begin{compactenum}
  \item $\langle \mathtt{div}\ \texttt{R1}, \texttt{R2}, \texttt{R3} \mid \sigma_1 \rangle$ \emph{vs.} $\langle \mathtt{div}\ \texttt{R1}, \texttt{R2}, \texttt{R3} \mid \sigma_4 \rangle$, for $\sigma_4$ with $\sigma_4(\texttt{R3}) = 8$, since $\mathtt{div}$ will take 33 cycles in both cases, and
  \item $\langle \mathtt{div}\ \texttt{R1}, \texttt{R2}, \texttt{R3} \mid \sigma_2 \rangle$ \emph{vs.} $\langle \mathtt{li}\ \texttt{R1}, \mathtt{0x0} \mid \sigma_3 \rangle$, since executing both $\mathtt{div}$ (since $r_3$ is $0$) and  $\mathtt{li}$ take 2 cycles.
\end{compactenum}

\mypara{Leakage contracts}
Securely programming $\uarch$ requires understanding which program executions an attacker might distinguish to ensure that secret data is not leaked.
Leakage contracts~\cite{contracts2021,Wang23} provide a way of characterizing such leaks at ISA-level, thereby enabling secure programming.
For this, contracts map architectural executions, \ie executions according to the $\arch$, to \emph{contract traces} that capture what information might be leaked through side channels.\looseness=-1

A contract is \emph{sound} for a CPU (or, equivalently, the CPU satisfies the contract), if the contract captures \emph{all} leaks in the  CPU, that is, any two program executions that are distinguishable by $\atk$ must result in different contract traces.
For instance, $\uarch$ satisfies a contract that exposes all the operands of executed $\mathtt{div}$ instructions, since distinguishable executions like the ones mentioned above would be mapped to different contract traces. 
While sound contracts are a prerequisite for secure programming, not all sound contracts are useful.
Indeed, consider a contract that exposes the operands of \emph{all} instructions.
Despite being sound for $\uarch$, this contract is useless from a secure programming perspective, since it only admits programs that do not process any secrets! %
This contract is imprecise since it over-approximates the actual leaks in $\uarch$. %

\subsection{Synthesizing sound and precise contracts with \tool{}}\label{sec:overview:synthesis}
\Cref{fig:approach-overview} depicts \tool{}'s approach  to synthesize sound and precise  contracts given an RTL processor. %
The approach takes as input:
\begin{compactenum}[(1)]
\item The RTL design of the processor under synthesis $\uarch$.
\item A user-provided \emph{contract template} $\template$ consisting of a set of \emph{contract atoms} (\Cref{sec:contracts}), each capturing potential instruction-level leaks. 
The template defines the search space for contract synthesis and any subset of its atoms is a candidate contract.\looseness=-1
\end{compactenum}
Based on these inputs, \tool{} outputs a contract $\ctr$ that is (a) \emph{sound}, \ie it captures \emph{all} leaks in the target CPU $\uarch$, and (b) \emph{as precise as possible}, \ie among all sound contracts in the  template it distinguishes the fewest of the attacker-indistinguishable test cases explored during synthesis.
Next, we describe the contract template for our  example, and then the three phases of our approach.\looseness=-1

\mypara{Contract template}
We consider a contract template that can expose as part of the contract trace three kinds of information: (1) the value $\myvalue{\imm}$ of the immediate $\mathtt{imm}$ in load-immediate instructions, (2) the destination and source registers used in $\divT$ instructions (denoted as $\myvalue{\texttt{RD}}, \myvalue{\texttt{RS1}}, \myvalue{\texttt{RS2}}$), and (3) the values of these registers (denoted as $\myvalue{\texttt{Reg[RD]}}$, $\myvalue{\texttt{Reg[RS1]}}$, and $\myvalue{\texttt{Reg[RS2]}}$).

Furthermore, an observation corresponding to the instruction type $\itype = \{\liT, \divT\}$ can be added to the trace whenever an instruction of that type is executed. 

Therefore, the set of atoms in our template is as follows:
\begin{align*}
  \atoms\ =\ \hspace{-5mm} &\hspace{5mm} \{(\liT, \langle\texttt{"\liT"}, -\rangle), (\divT, \langle\texttt{"\divT"}, -\rangle), (\liT, \langle\texttt{"\imm"}, \myvalue{\imm}\rangle), \\
    & (\divT, \langle\texttt{"RD"}, \myvalue{\texttt{RD}}\rangle), (\divT, \langle\texttt{"Reg[RD]"}, \myvalue{\texttt{Reg[RD]}}\rangle), \\
    & (\divT, \langle\texttt{"RS1"}, \myvalue{\texttt{RS1}}\rangle), (\divT, \langle\texttt{"Reg[RS1]"}, \myvalue{\texttt{Reg[RS1]}}\rangle), \\ 
    & (\divT, \langle\texttt{"RS2"}, \myvalue{\texttt{RS2}}\rangle), (\divT, \langle\texttt{"Reg[RS2]"}, \myvalue{\texttt{Reg[RS2]}}\rangle)\}
\end{align*} 
For instance, atom $(\liT, \langle\texttt{"\liT"}, -\rangle)$ adds to the trace the tuple\linebreak $\langle \texttt{"\liT"}, -\rangle$ whenever a load immediate instruction is executed, \linebreak whereas atom $(\divT, \langle\texttt{"Reg[RS2]"}, \myvalue{\texttt{Reg[RS2]}}\rangle)$ exposes the tuple $\langle \texttt{"Reg[RS2]"},$ $\myvalue{\texttt{Reg[RS2]}} \rangle$ containing the value of the second source operand whenever a $\divT$ instruction is executed.

\mypara{Phase 1: Empirical leakage characterization} 
To empirically characterize a processor's leakage, \tool{} first generates a set of test cases  $T$, where each test case consists of a pair of ISA-level programs with associated data inputs.
Then, it simulates the execution of the test cases on the target processor~$\uarch$ and determines which test cases are \emph{attacker distinguishable} and which are not.

To allow the subsequent synthesis step to associate leakage with particular contract atoms in the template, \tool{} also simulates each test case on the contract template $\template$.
The result of this simulation is condensed into the \emph{template distinguishability} of each test case, which compactly captures which sets of contract atoms distinguish a particular test case and which do not.

In our example, we consider an attacker $\atk$ that observes when instructions retire and we systematically generate the following test cases, each consisting of a pairs of program executions:
\begin{compactitem}
  \item[($T_1$)] $\langle\liT\ \texttt{R1}, \texttt{0x1234} \mid \sigma\rangle$ \textit{vs.} $\langle\liT\ \texttt{R1}, \texttt{0x5678} \mid \sigma\rangle$ for some $\sigma$.
  \item[($T_2$)] $\langle\liT\ \texttt{R1}, \texttt{0x1234} \mid \sigma\rangle$ \textit{vs.} $\langle\liT\ \texttt{R2}, \texttt{0x1234} \mid \sigma\rangle$ for some $\sigma$.
  \item[($T_3$)] $\langle\liT\ \texttt{R1}, \texttt{0x1234} \mid \sigma\rangle$ \textit{vs.} $\langle\liT\ \texttt{R1}, \texttt{0x1234} \mid \sigma'\rangle$ for $\sigma, \sigma'$, s.t. 
  $\sigma(\texttt{R1}) \neq \sigma'(\texttt{R1})$, \ie for different initial values of \texttt{R1}.
  \item[($T_4$)] $\langle\liT\ \texttt{R1}, \texttt{0x1234} \mid \sigma\rangle$ \textit{vs.} $\langle\divT\ \texttt{R1}, \texttt{R2}, \texttt{R3} \mid \sigma\rangle$ for some~$\sigma$.\looseness=-1
  \item[($T_5$)] $\langle\divT\ \texttt{R1}, \texttt{R2}, \texttt{R3} \mid \sigma\rangle$ \textit{vs.} $\langle\divT\ \texttt{R1}, \texttt{R2}, \texttt{R3} \mid \sigma'\rangle$ for $\sigma, \sigma'$, s.t.
    $\sigma(\texttt{R1}) \neq \sigma'(\texttt{R1})$, $\sigma(\texttt{R2}) \neq \sigma'(\texttt{R2})$, and $\sigma(\texttt{R3}) \neq \sigma'(\texttt{R3})$. Further, $\sigma(\texttt{R3}),\sigma'(\texttt{R3}) \not\in \{0, 1\}$, \ie neither of the initial states assigns $0$ or $1$ to \texttt{R3}. 
\end{compactitem}
For the attacker, all test cases except $T_4$ are indistinguishable, since the program executions take the same amount of time.
In terms of template distinguishability, for instance, test case $T_4$ is distinguished by any contract that includes at least one of the template's atoms.
   
\mypara{Phase 2a: Contract Synthesis}  
Based on the empirical leakage characterization,  \tool{}  synthesizes, using integer linear programming, a candidate contract $\ctr$  
that (a) distinguishes all attacker-distinguishable test cases, and (b) is as precise as possible given the template $\template$. In other words, it distinguishes as few of the attacker-indistinguishable test cases as possible among all contracts in $\template$.

In our example, based on the initial test cases, \tool{} might synthesize the following initial contract:
\begin{align*}
  \ctr_1 = \{&(\liT, \langle\texttt{"\liT"}, -\rangle)\}
\end{align*}
This contract $\ctr_1$ captures that the attacker $\atk$ can distinguish  \liT\ instructions  from  \divT\ instructions (as indicated by $T_4$ above).
Note that $\ctr_1$ is the most precise contract that captures all leaks exercised by the test cases.
For instance, the contract exposing also the value of the immediate in $\liT$ instructions, \ie $(\liT, \langle\texttt{"\imm"}, \myvalue{\imm}\rangle)$, would still capture the leak in $T_4$, but it also unnecessarily distinguishes $T_1$ and it is, therefore, not a possible solution.

\mypara{Phase 2b: Bounded verification of soundness}  
Since the test cases may miss some leaks, the synthesized contract is not guaranteed to be sound.
To ensure soundness, \tool{} alternates the contract synthesis step with a bounded verification step, which either provides a counterexample to soundness or it proves that the contract is sound up to a given bound on the number of executed cycles.
In case of unsoundness, the counterexample is used to refine the contract by looping back to the contract synthesis step (\emph{Phase 2a}).\looseness=-1

In our example, $\ctr_1$ is unsound since it does not capture the leaks generated by \divT\ instructions.
The bounded verification might generate the following test case as a counterexample to soundness: 
\begin{compactitem}
  \item[($T_6$)] $\langle\divT\ \texttt{R1}, \texttt{R2}, \texttt{R3} \mid \sigma\rangle$ \emph{vs.} $\langle\divT\ \texttt{R1}, \texttt{R2}, \texttt{R3} \mid \sigma'\rangle$ for $\sigma, \sigma'$, s.t. $\sigma(\texttt{R3}) = 0$ and $\sigma'(\texttt{R3}) = 5$.
\end{compactitem}
Next, \tool{} simulates the contract template $\template$ on the newly generated counterexample to determine which of the template's contracts would distinguish it.
Then, it re-executes the contract synthesis step (\emph{Phase 2a}).
This time \tool{} generates the following contract:
\begin{align*}
  \ctr_2 = \{&(\divT, \langle\texttt{"Reg[RS2]"}, \myvalue{\texttt{Reg[RS2]}}\rangle)\}
\end{align*}
$\ctr_2$  captures that the value of the second source register influences the execution time of the \divT\ instruction, thereby distinguishing~$T_6$.
The presence of an atom that applies only to the \divT\ instruction also ensures that \divT\ instructions are distinguishable from \liT\ instructions, which distinguishes $T_4$.

\mypara{Phase 3: Unbounded verification of soundness}
As soon as the bounded verification does not find any more counterexamples, \tool{} attempts to verify the synthesized contract unboundedly.
To this end, \tool{} generates inductive invariants using the Houdini algorithm, following the verification approach of Wang et al.~\cite{Wang23}.
If the verification succeeds, \tool{} returns the synthesized sound and precise contract to the user.
In our example, $\ctr_2$ can be proved sound in an unbounded manner and \tool{} terminates.

\subsection{Generating more precise contracts}\label{sec:overview:precision}

\tool{} synthesizes the most precise sound contract among all possible contracts given a specific contract template.
The choice of the template, however, can affect the synthesized contracts' precision.

The sound contract $\ctr_2$ from \Cref{sec:overview:synthesis} is imprecise.
For instance, it distinguishes the test case below, even though the CPU $\uarch$ executes both instructions in the same number of cycles:
\begin{compactitem}
  \item[($T_7$)] $\langle\divT\ \texttt{R1}, \texttt{R2}, \texttt{R3} \mid \sigma\rangle$ \textit{vs.} $\langle\divT\ \texttt{R1}, \texttt{R2}, \texttt{R3} \mid \sigma'\rangle$ for $\sigma, \sigma'$, s.t. $\sigma(\texttt{R3}) = 7$ and $\sigma'(\texttt{R3}) = 8$.
\end{compactitem}
However, \tool{} cannot generate a more precise contract since the template from \Cref{sec:overview:synthesis} can only expose the entire value of operands for $\divT$ instructions.

To derive more precise contracts, we need a more expressive contract template, which can be achieved by including additional atoms of the following form, where $k$ is a constant value:
$$(\divT{}, \langle\texttt{"Reg[RS2]=k?"}, \myvalue{\texttt{Reg[RS2]}} == k\rangle)$$ 

Running \tool{} again with the updated template and test cases $T_1, \ldots, T_7$ may result in generating the following contract at the end of \emph{Phase 2a}:
\begin{align*}
  \ctr_2 = \{&(\divT, \langle\texttt{"Reg[RS2]=0?"}, \myvalue{\texttt{Reg[RS2]}} == 0\rangle)\}
\end{align*}
Even though this contract captures one of the leaks caused by $\divT$, it is not sound since it misses the other special case of the \divT\ instruction, \ie when \texttt{R2} is $1$.
When running the bounded verification (\emph{Phase 2b}), \tool{} might produce the following counterexample: 
\begin{compactitem}
  \item[($T_8$)] $\langle \divT\ \texttt{R1}, \texttt{R2}, \texttt{R3} \mid \sigma\rangle$ \emph{vs.} $\langle\divT\ \texttt{R1}, \texttt{R2}, \texttt{R3} \mid \sigma'\rangle$\\ for $\sigma, \sigma'$, s.t.
  $\sigma(\texttt{R2}) = 1$ and $\sigma'(\texttt{R2}) = 5$.
\end{compactitem}
A new synthesis step, then, finally results in the following contract:
\begin{align*}
  \ctr_3 = \{&(\divT, \langle\texttt{"Reg[RS2]=0?"}, \myvalue{\texttt{Reg[RS2]}} == 0\rangle), \\
  &(\divT, \langle\texttt{"Reg[RS2]=1?"}, \myvalue{\texttt{Reg[RS2]}} == 1\rangle)\}
\end{align*}
This contract is sound and, again, it is the most precise given the updated template.
Therefore, \tool{} terminates by returning $\ctr_3$ to the user after the final unbounded verification.

\section{Formal model}\label{sec:formal-model}

Here, we present the core components of our formal model.

\subsection{Architectures and microarchitectures}

Leakage contracts act as a security abstraction between the instruction set architecture (short: architecture or ISA) and a microarchitecture, that is, a concrete implementation in a processor.
Next, we formalize both concepts and what it means for a microarchitecture to correctly implement an architecture~\cite{Wang23}.

\mypara{Architectures}
We view an architecture as a state machine that defines how the execution progresses through a sequence of architectural states, where each transition corresponds to the execution of a single instruction.
Formally, an \emph{architecture} is a pair $(\archStates, \archStep)$ where $\archStates$ is a set of architectural states (each modeling the values of registers and data/instruction memory) and $\archStep: \archStates \to \archStates$ is a transition function that maps each state $\sigma \in \archStates$ to its successor $\archStep(\sigma)$, obtained by executing the next instruction to be executed in $\sigma$.
To capture an execution, we denote by $\archEval(\sigma)$ the sequence of states reached from $\sigma$ by successive applications of $\archStep$, \ie $\archEval(\sigma) = \sigma_0, \sigma_1, \sigma_2, ...$ where $\sigma_0 = \sigma$ and $\sigma_{i+1} = \archStep(\sigma_i)$ for all $i \geq 0$.

\mypara{Microarchitectures}
We view microarchitectures as state machines modeling how the processor's state evolves  at cycle level.
Thus, a \emph{microarchitecture} is a triple $(\implStates, \initImplStates, \implStep)$ where
 $\implStates = \archStates \times \uarchStates$ is the set of microarchitectural states (some of which are initial states $\initImplStates \subseteq \implStates$) 
 and $\implStep: \implStates \to \implStates$ is a transition function that maps each   state $\sigma \in \implStates$ to its successor $\implStep(\sigma)$, obtained by executing the processor for one cycle.
Each microarchitectural state $\sigma \in \implStates$ consists of an architectural part $\archProj{\sigma} \in \archStates$ and a microarchitectural part $\uarchProj{\sigma} \in \uarchStates$ modeling the state of microarchitectural components (such as caches and predictors).
For simplicity, we require that there is a canonical initial microarchitectural component $\mu_0 \in \uarchStates$ such that  $\initImplStates = \{ (\sigma, \mu_0) \mid \sigma \in \archStates\}$.
Similarly to  $\archEval(\sigma)$,  $\implEval(\sigma)$ denotes the sequence of states reached from $\sigma$ by successive applications of $\implStep$.
Finally, given a predicate $\phi$ over $\implStates$, $\implEvalFilter{\phi}(\sigma)$ denotes the sequence of elements from $\implEval(\sigma)$ that satisfy $\phi$.
For simplicity, we refer to a microarchitecture only using its transition function $\implStep$ rather than the  tuple $(\implStates, \initImplStates, \implStep)$.
Similarly, we refer to an architecture $(\archStates, \archStep)$ simply using its transition function $\archStep$.

\mypara{ISA compliance}
To correctly implement an architecture $\arch$, an implementation $\uarch$ needs to change the architectural state in a manner consistent with $\arch$.
We capture this following the ISA compliance notion from Wang et al.~\cite{Wang23}, which relies on a \emph{retirement predicate} $\phi$ that indicates when $\uarch$ retires instructions.
Then, we say that a microarchitecture $\uarch$ implements an architecture $\arch$ if one can map changes of the architectural state in $\uarch$ to $\arch$'s executions using $\phi$. %
 
Formally, a microarchitecture $\uarch$ \emph{correctly implements} an architecture $\arch$~\cite{Wang23} given a retirement predicate $\phi$  over $\archStates$, written $\isasat{\uarch}{\arch}{\phi}$, if for all states $\sigma \in \initImplStates$:
\begin{compactenum}[(a)]
\item $\archProj{(\implEvalFilter{\phi}(\sigma))} = \archEval( \archProj{\sigma} )$, \ie all architectural changes witnessed by $\phi$ agree with $\arch$, and 
\item $\archProj{(\implStep^i(\sigma))} = \archProj{(\implStep^{i-1}(\sigma))}$ whenever $\archProj{(\implStep^i(\sigma))} \\ \not\models \phi$, \ie  no architectural changes beyond those witnessed by $\phi$.\looseness=-1 
\end{compactenum}

\subsection{Leakage contracts}\label{sec:contracts}

Next, we introduce how we model leakage contracts in \tool{}.
We first introduce \emph{contract atoms}, which are the basic building block for contracts. %
Next, we introduce  \emph{contract templates}, which capture the synthesis space for our approach as a set of possible atoms.

\mypara{Contract atoms}
Atoms capture potential leaks at the instruction level, and they  are built from \emph{applicability predicates} and \emph{leakage functions}.
The former determine whether a contract atom is applicable in a given architectural state, whereas the latter determine what information from the architectural state is leaked in case an atom is applicable.
Formally, a \emph{contract atom} $A$ is a pair $A = (\pi_A, \phi_A)$ where the applicability predicate $\pi_A$ is a predicate over $\archStates$ and the leakage function $\phi_A$ is a function $\phi_A : \archStates \to \ctrObs$ mapping an architectural state to a contract observation in $\ctrObs$.\looseness=-1

\begin{example}\label{example:template}
Addition and subtraction instructions $\mathtt{addi}$ and $\mathtt{subi}$ could both leak their immediate value.
Contracts atoms for these instructions can be defined as follows:
\begin{asparaitem}
    \item $(\texttt{addi}, \langle\texttt{"imm"}, \myvalue{\texttt{imm}}\rangle)$, where predicate $\texttt{addi}$ holds whenever the instruction is an addition instruction and $\langle\texttt{"imm"}, \myvalue{\texttt{imm}}\rangle$ is a leakage function that leaks the immediate value of the instruction $\myvalue{\texttt{imm}}$ together with a leakage identifier $\texttt{"imm"}$.
    \item $(\texttt{subi}, \langle\texttt{"imm"}, \myvalue{\texttt{imm}}\rangle)$, where predicate $\texttt{subi}$ holds whenever the instruction is a subtraction instruction and the leakage function is the same as above.
\end{asparaitem}
Note that  $\langle\texttt{addi}\ \texttt{R1}, \texttt{R2}, \texttt{0x1234} \mid \sigma\rangle$ \textit{vs.} $\langle\texttt{subi}\ \texttt{R1}, \texttt{R2}, \texttt{0x1234} \mid \sigma\rangle$ would be indistinguishable by a contract including both atoms.
Even though two different instructions are executed, both atoms would be applicable and expose $\langle\texttt{"imm"}, \texttt{0x1234} \rangle$.
\end{example}

\mypara{Contract templates}
A \emph{contract template}~$\template$ is a set of contract atoms.
We require that the applicability predicates of contract atoms with the same leakage function are mutually exclusive, \ie at most one of these atoms is applicable in any given state.
This holds in \Cref{example:template}, since instructions can either be  additions or subtractions.

We also require that the images of different leakage functions are disjoint.
This ensures that leakage from one function cannot ``cancel out'' leakage from another function. 
For instance, a leakage function exposing an immediate operand and another  function exposing the instruction's opcode must not leak the same value.
A simple way to satisfy this requirement is associating a unique identifier with each leakage function, like \texttt{"imm"} in \Cref{example:template}.

\mypara{Leakage contracts}
Given a template $\template$, any subset $S \subseteq \template$ of the template induces a \emph{contract} $\contract_S : \archStates \to \powerset{\ctrObs}$, which is a function from architectural states to sets of observations.\footnote{We often refer to a contract $\contract_S$ directly using the set of atoms inducing it.}
Formally, $\contract_S$ is defined by evaluating each applicable atom in $S$ as follows:
$\contract_S(\sigma) := \{ \phi_A(\sigma) \mid  (\pi_A, \phi_A) \in S \wedge \pi_A(\sigma) \}$.
Given a sequence $\tau:=\sigma_0, \sigma_1, \ldots$ of architectural states, $\contract_S(\tau)$ denotes the corresponding sequence of observations $\contract_S(\sigma_0), \contract_S(\sigma_1), \ldots$.
Therefore, given an architectural state~$\sigma \in \archStates$, $\contract_S(\archEval(\sigma))$ denotes the \emph{contract trace}, \ie the sequence of contract observations associated with the execution $\archEval(\sigma)$.

Finally, in \tool{}, a \emph{test case} $T = (\sigma, \sigma')$ is a pair of architectural states. %
We say that $T$ is \emph{contract distinguishable} for given contract $\contract_S$ if the corresponding contract traces are different, \ie $\contract_S(\archEval(\sigma)) \neq \contract_S(\archEval(\sigma'))$.

\subsection{Contract satisfaction}

We conclude by formalizing contract satisfaction~\cite{contracts2021} in our setting.
For this, we  model microarchitectural attackers and then characterize when a contract captures all leaks observable by the attacker.

\mypara{Attackers}
We consider (passive) microarchitectural attackers that can  extract information from the microarchitectural state.
Formally, we model a microarchitectural attacker as a function $\atk: \implStates \to \atkObs$ that maps microarchitectural states  to attacker observations in $\atkObs$.
Common attacker models, like the one exposing the timing of instruction retirement~\cite{Tsunoo03} or the one exposing the final state of caches~\cite{Yarom14,Doychev2015}, can be instantiated in this setting.
Given an initial microarchitectural state~$\sigma \in \initImplStates$, $\atk(\implEval(\sigma))$ denotes the \emph{attacker trace}, \ie the sequence of attacker observations associated with the cycle-accurante microarchitectural execution $\implEval(\sigma)$.
We say that a test case $T = (\sigma, \sigma')$ is \emph{attacker distinguishable} if the corresponding attacker traces are different, \ie $\atk(\implEval(\sigma)) \neq \atk(\implEval(\sigma'))$.

\mypara{Contract Satisfaction}
A microarchitecture $\uarch$ {satisfies} the contract $\contract_S$ for an attacker $\atk$ if $\atk$ cannot learn more information about the initial architectural state by monitoring $\uarch$'s executions than what is exposed by the contract.
That is, for any two initial states,
whenever the contract traces are the same, then the attacker traces must be identical, \ie{} $\atk$ cannot distinguish the two architectural executions.
This is formalized in \Cref{def:contract-satisfaction}.

\begin{definition}\label{def:contract-satisfaction}
Microarchitecture $\uarch$ \emph{satisfies} contract $\contract_S$ for attacker $\atk$, written $\ctrsat{ \contract_S }{ \atk }$, if for all initial states $\sigma,\sigma' \in \initImplStates$, 
    if  $\contract_S(\archEval(\archProj{\sigma})) = \contract_S(\archEval(\archProj{\sigma'}))$, 
    then $\atk(\implEval(\sigma)) = \atk(\implEval(\sigma'))$.
\end{definition}

\Cref{def:contract-satisfaction} 
 refers to 4 different traces: two contract traces defined over the architecture $\arch$ and two attacker traces defined over the microarchitecture $\uarch$.
Wang et al.~\cite{Wang23} proposed the notion of {microarchitectural contract satisfaction}, formalized below, which is expressed only over a pair of microarchitectural traces and allows to decouple reasoning about leakage and about ISA compliance.

\begin{definition}\label{def:uarch-contract-satisfaction}
    Microarchitecture $\uarch$ \emph{microarchitecturally-satisfies} contract $\contract_S$ for attacker $\atk$ and predicate $\phi$, written $\uarchctrsat{ \contract_S }{ \atk }{\phi}$, if for all initial states $\sigma,\sigma'$, 
    if  $\contract_S(\implEvalFilter{\phi}(\sigma)) = \contract_S(\implEvalFilter{\phi}(\sigma'))$, 
    then $\atk(\implEval(\sigma)) = \atk(\implEval(\sigma'))$.
\end{definition}

As stated in \cite[Theorem 1]{Wang23}, whenever  $\uarch$ correctly implements  $\arch$, then contract satisfaction and microarchitectural contract satisfaction are equivalent.
Hence, \Cref{def:uarch-contract-satisfaction} is the notion targeted by \tool{} and, for conciseness, we will refer to it simply as ``contract satisfaction'' throughout the rest of the paper.

\section{Contract synthesis}\label{sec:synthesis}

In this section, we describe \tool{}'s synthesis approach. %
First, we describe the empirical leakage characterization (\emph{Phase 1} in Fig.~\ref{fig:approach-overview}; \Cref{sec:synthesis:testing}).
Next, we describe how \tool{} synthesizes a candidate contract using integer linear programming (\emph{Phase 2a} in Fig.~\ref{fig:approach-overview}; \Cref{sec:synthesis:synthesis}).
Then, we describe how \tool{} checks whether a contract is sound, \ie it capture all leaks, with bounded verification (\emph{Phase 2b} in Fig.~\ref{fig:approach-overview}; \Cref{sec:synthesis:bounded-verification}) and unbounded verification (\emph{Phase 3} in Fig.~\ref{fig:approach-overview}; \Cref{sec:synthesis:unbounded-verification}).
The proofs of all propositions are given in \techReportAppendix{appendix:proofs}.

\subsection{Phase 1: Empirical leakage characterization}\label{sec:synthesis:testing}

\tool{} generates a set of test cases $T$ to characterize a processor's leakage.
As an input to the ILP-based contract synthesis (\emph{Phase 2a}), \tool{} needs to determine for each test case:
\begin{inparaenum}[(a)]
    \item whether it is \emph{attacker distinguishable}, and
    \item which contracts from the contract template would make it \emph{contract distinguishable}.
\end{inparaenum}

\mypara{Test-case generation}
Test cases aid contract synthesis in two opposing ways:
\begin{inparaenum}[(1)]
\item Attacker-distinguishable test cases uncover leaks and force the {inclusion} in the contract of atoms that expose these leaks.
\item Attacker-indistinguishable test cases show which atoms should {not} be included in the contract to avoid imprecision.
\end{inparaenum}

Based on these observations, our goal is to generate test cases that differ in exactly one atom.
This has two benefits.
First, if a test case is attacker distinguishable, there is only one possible responsible atom.
Second,  as test cases ``differ'' only in  one atom, this strategy is likely to generate many attacker indistinguishable tests, in particular when targeting atoms that do not capture leakage in the CPU.\looseness=-1

To generate such test cases, we proceed as follows.
We start by generating an architectural state $\sigma_0$ that consists of a random instruction sequence and a random initial valuation of the registers.
Then, we randomly pick an instruction from the sequence and an applicable leakage function and generate a pair of architectural states $(\sigma,\sigma')$ such that the leakage function will evaluate differently in $\sigma$ and $\sigma'$ on the chosen instruction.

To implement this, we define a \emph{modifier function} for each leakage function in our template, similarly to~\cite{Mohr24}.
Modifier functions may, for instance, change the initial value of a register, inject new  instruction to adapt the architectural state, or modify an instruction.\looseness=-1

Recall the template introduced in \Cref{sec:overview:synthesis}. 
Given this template, the modifier function for the leakage function $\imm$ would for example transform a state $\sigma_0$ whose first instruction is \texttt{li}\ \texttt{R1}, \texttt{0x1234} to the pair $(\sigma, \sigma')$ where $\sigma = \sigma_0$ and the only difference in $\sigma'$ is that its first instruction is \texttt{li}\ \texttt{R1}, \texttt{0x5678}.
Similarly, the modifier function for the leakage function \texttt{RD} could substitute the instruction \texttt{div}\ \texttt{R1}, \texttt{R2}, \texttt{R3} in $\sigma_0$ with \texttt{div}\ \texttt{R4}, \texttt{R2}, \texttt{R3} in $\sigma'$.
For other leakage functions such as \texttt{Reg[RS1]}, the modifier function does not alter any instructions but modifies the value of a specific register in the initial architectural state of $\sigma'$.
Depending on the context of an instruction, modifier functions are not guaranteed to generate a test case that only differs in the chosen leakage function, \eg if dependent instructions are present in the instruction sequence. 

\mypara{Attacker distinguishability}
Each of the test cases $t = (\sigma, \sigma')$ in $T$ is executed on the target CPU $\uarch$.
To derive the attacker traces for $\sigma$ and $\sigma'$, we construct the two initial states $(\sigma, \mu_0)$ and $(\sigma',\mu_0)$ for the CPU $\uarch$.
To determine attacker indistinguishability, we then simulate the executions $\implEval((\sigma, \mu_0))$ and $\implEval((\sigma',\mu_0))$ and check whether the attacker can distinguish them, \ie we check if $\atk(\implEval((\sigma, \mu_0))) \neq \atk(\implEval((\sigma',\mu_0)))$.
In general, the executions $\implEval((\sigma, \mu_0))$ and $\implEval((\sigma',\mu_0))$ might be infinite, so \tool{} simulates them up to a fixed number of cycles $n$, and we pick an $n$ that is large enough to fully simulate all our test cases. %

\mypara{Template distinguishability}\label{sec:synthesis:template-distinguishability}
To succinctly characterize the contract distinguishability of each test case for all possible contracts induced by the template $\template$, we rely on the following observation.

Let $(\sigma, \sigma') \in T$ be a test case and let $\contract_S$ be a contract.
Also, let $\archEval (\sigma) = \sigma_0, \sigma_1, \ldots $ and $\archEval (\sigma') = \sigma_0', \sigma_1', \ldots $ be the architectural executions from $\sigma$ and $\sigma'$.
Test case $(\sigma, \sigma')$ can be contract distinguishable under $\contract_S$ for two distinct reasons:
\begin{asparaitem}
    \item \textit{Leakage mismatch:} There are atoms $A,B \in S$ with $\phi_A = \phi_B$ and an index $i$ s.t. $\pi_A(\sigma_i) \wedge \pi_B(\sigma'_i)$ and $\phi_A(\sigma_i) \neq \phi_B(\sigma_i')$, \ie at some point, the atoms $A$ and $B$ are applicable in the two executions but evaluate differently.
        Note that $A$ and $B$ can be the same atom.
    \item \textit{Applicability mismatch:} There is an atom $A \in S$ and an index~$i$ s.t. $\pi_A(\sigma_i)$ (or $\pi_A(\sigma'_i)$) and no atom $B \in S$ with $\phi_A = \phi_B$ is applicable in $\sigma'_i$ (or $\sigma_i$), \ie there is an observation only in one  executions.
\end{asparaitem}

\begin{example}\label{example:approach:empirical}
    Consider the following test cases:
    \begin{asparaitem}
        \item Test case  \texttt{li}\ \texttt{R1}, \texttt{0x1234} \textit{vs.} \texttt{li}\ \texttt{R1}, \texttt{0x5678} would be distinguishable due to a leakage mismatch by a contract including the atom $(\texttt{li}, \langle\texttt{"imm"}, \myvalue{\texttt{imm}}\rangle)$, since the immediate value exposed by the atom is different in the two executions.
        \item Test case \texttt{li}\ \texttt{R1}, \texttt{0x1234} \textit{vs.} \texttt{add}\ \texttt{R1}, \texttt{R2}, \texttt{R3} would be distinguishable due to an applicability mismatch by a contract consisting only of the atom $(\texttt{li}, \langle\texttt{"imm"}, \myvalue{\texttt{imm}}\rangle)$ as no atom with the same leakage function is applicable in the second execution.
    \end{asparaitem}
\end{example}

We characterize the template distinguishability of each test case $t = (\sigma, \sigma')$ via its \emph{strongly-distinguishing atoms} $\mathit{SD}_t \subseteq \template$ and its \emph{xor-distinguishing pairs} $\mathit{XOR}_t \subseteq \binom{\template}{2}$:

A \emph{strongly-distinguishing atom} $A \in \template$ is one for which there exists an index $i$ such that $\pi_A(\sigma_i)$ (or $\pi_A(\sigma'_i)$) holds, and no atom $B \in \template$ satisfies $\pi_B(\sigma'_i) \wedge \phi_A(\sigma_i) = \phi_B(\sigma'_i)$ (or $\pi_B(\sigma_i) \wedge \phi_B(\sigma_i) = \phi_A(\sigma'_i)$).
Including any strongly-distinguishing atom $A \in \mathit{SD}_t$ in a contract will make the test case $t$ contract distinguishable independently of which other atoms are included.
This contract distinguishability may be caused by a leakage mismatch \emph{or} an applicability mismatch.

For example, in the test case \liT\ \texttt{R1}, \texttt{0x1234} \textit{vs.} \liT\ \texttt{R1}, \texttt{0x5678} the atom $(\texttt{li}, \langle\texttt{"imm"}, \myvalue{\texttt{imm}}\rangle)$ is strongly distinguishing, as it is applicable in both executions and the immediate values are different.

A \emph{xor-distinguishing atom pair} $\{A, B\} \in \binom{\template}{2}$ with $A \neq B$ is one for which there exists an index $i$ such that $\pi_A(\sigma_i) \wedge \pi_B(\sigma'_i)$ holds and $\phi_A(\sigma_i) = \phi_B(\sigma'_i)$.
Due to our assumption that the applicability predicates of contract atoms that share a leakage function are mutually exclusive, there can be no other atom $C \in \template$ with $\phi_A = \phi_B = \phi_C$ that is applicable in $\sigma_i$ or $\sigma'_i$.
Thus, including exactly one of the two atoms $A$ and $B$ in a xor-distinguishing pair $\{A, B\} \in \mathit{XOR}_t$ will make the test case $t$ contract distinguishable due to an applicability mismatch.
For example, in the test case \texttt{addi}\ \texttt{R1}, \texttt{R2}, \texttt{0x1234} \textit{vs.} \texttt{subi}\ \texttt{R1}, \texttt{R2}, \texttt{0x1234} the pair of atoms $(\texttt{addi}, \langle\texttt{"imm"}, \myvalue{\texttt{imm}}\rangle)$ and $(\texttt{subi}, \langle\texttt{"imm"}, \myvalue{\texttt{imm}}\rangle)$ is a xor-distinguishing pair.

The two sets $\mathit{SD}_t$ and $\mathit{XOR}_t$ fully characterize the contract distinguishability of a test case $t$ under any possible contract: %

\begin{restatable}{proposition}{templatedistinguishability}\label{thm:template-distinguishability}
    Let $t$ be a test case and let $\contract_S$ be a contract.
    Then, $t$ is contract distinguishable under $\contract_S$ iff
    \begin{compactitem}[\hspace{1em}$\bullet$]
        \item $S$ includes a strongly-distinguishing atom, \ie $\mathit{SD}_t \cap S \neq \emptyset$, or
        \item $S$ includes exactly one of the two atoms in a xor-distinguishing pair, \ie $\exists \{A, B\} \in \mathit{XOR}_t$ with $|S \cap \{A, B\}| = 1$.
    \end{compactitem}
\end{restatable}

To compute the sets $\mathit{SD}_t$ and $\mathit{XOR}_t$ for a test case $t$, we simulate the execution of the test case w.r.t. the contract template $\template$ and check the conditions for strongly-distinguishing and xor-distinguishing atoms for each index $i$.
To simulate $\archEval (\sigma)$ (which might be infinite), \tool{} only simulates executions up to a fixed number of steps.\looseness=-1 %

\revision{
\mypara{Relation to Mohr et al.~\cite{Mohr24}}
Next, we discuss how our characterization of template distinguishability in terms of strongly-distinguishing atoms and xor-distinguishing pairs relates to the one introduced by Mohr et al.~\cite{Mohr24}, which relies on the notion of \emph{distinguishing atoms}.
An atom $A \in \template$ is \emph{distinguishing}~\cite{Mohr24} if the contract consisting solely of atom $A$ distinguishes the test case, \ie if $\contract_{\{A\}}(\archEval(\sigma)) \neq \contract_{\{A\}}(\archEval(\sigma'))$.
Equivalently, $A$ is distinguishing if there exists an index~$i$ such that $\pi_A(\sigma_i) \oplus \pi_A(\sigma'_i)$ or 
$\pi_A(\sigma_i) \wedge \pi_A(\sigma'_i) \wedge \phi_A(\sigma_i) \neq \phi_A(\sigma'_i)$.

For templates that contain multiple atoms that share the same leakage function (like the templates implemented in \tool{} and used in \Cref{sec:evaluation}), the set of distinguishing atoms is \emph{not} sufficient to fully characterize a test case's contract distinguishability (as opposed to our characterization, as indicated by \Cref{thm:template-distinguishability}).
In particular, the inclusion of two distinguishing atoms may result in contract indistinguishability as \Cref{example:template} demonstrates.

In contrast, for templates where no two distinct atoms share the same leakage function, the set of distinguishing atoms is sufficient to characterize contract distinguishability and, therefore, our characterization is equivalent to the one from \cite{Mohr24}. 
\begin{restatable}{proposition}{templatedistinguishabilitydate}\label{thm:template-distinguishability-date}
    Let $t$ be a test case and let $\contract_S$ be a contract. 
    If no two distinct atoms share the same leakage function, then $t$ is contract distinguishable under $\contract_S$ iff $S$ includes a distinguishing atom for $t$.
\end{restatable}

To summarize, the approach by Mohr et al.~\cite{Mohr24} cannot handle templates with atoms that share the same leakage function (in addition to synthesizing contracts that are not guaranteed to be sound).
However, the templates used by \tool{} \emph{do} contain distinct atoms that share the same leakage functions and for this reason we require the more sophisticated characterization of contract distinguishability in terms of strongly-distinguishing atoms and xor-distinguishing pairs.
}

\subsection{Phase 2a: ILP-based contract synthesis}\label{sec:synthesis:synthesis}

Given a template $\template$, the empirical leakage characterization from \Cref{sec:synthesis:testing} computes the following information for each test case $t \in T$:
\begin{compactitem}
    \item $\mathit{SD}_t \subseteq \template$ the set of strongly-distinguishing atoms,
    \item $\mathit{XOR}_t \subseteq \binom{\template}{2}$ the set of xor-distinguishing atom pairs, and
    \item $d_t \in \mathbb{B}$ whether the test case is attacker distinguishable.
\end{compactitem}
For convenience, we denote the set of attacker-distinguishable test cases as $T_d = \{t \in T \mid d_t = 1\}$ and the set of attacker-indistinguishable test cases as $T_{nd} = T \setminus T_d$. %

Based on this information, \tool{} uses integer linear programming (ILP) to synthesize a contract.
That is, we compute a set of atoms $S \subseteq \template$, such that $\contract_S$ distinguishes all attacker-dist\-inguishable test cases and as few attacker-indistinguishable test cases as possible. %
Next, we detail this ILP formulation.

\mypara{Variables}
For each atom $A \in \template$, we introduce a boolean variable~$s_A$ that is true \emph{iff} the atom $A$ is included in the synthesized contract.

We also introduce a boolean variable $\textit{fp}_t$ for each $t \in T_{nd}$ to capture whether the attacker-indistinguishable test case $t$ is a false positive, \ie it is contract distinguishable in the synthesized contract.\looseness=-1

Finally, we introduce a boolean variable $x_{\{A,B\}}$ for each pair $\{A,B\}$ that is xor-distinguishing for some test case $t$. 

\mypara{Objective function}
The objective function of the ILP then is to minimize the number of false positives: $\min \sum_{t \in T_{nd}}{\textit{fp}_t}$.
As a secondary objective we minimize the number of atoms in the contract, \ie we minimize $\sum_{A \in \template}{s_A}$, which promises to yield a smaller contract that is likely more precise on unseen test cases.

\mypara{Constraints}
It remains to ensure that all attacker-distinguishable test cases are contract distinguishable and that the $\textit{fp}_t$ variables are consistent with the contract encoded by the $s_A$ variables.
From \Cref{thm:template-distinguishability}, we know that a test case $t$ is contract distinguishable \emph{iff} the contract contains (a) a strongly-distinguishing atom or (b) exactly one of the two atoms of one of its xor-distinguishing pairs.

We start by adding the following constraints enforcing $x_{\{A,B\}} = s_A \oplus s_B$ for any xor-distinguishing pair $\{A,B\} \in \bigcup_{t \in T_d} \mathit{XOR}_t$:
\begin{equation*}
\begin{aligned}\label{eq:xorconstraints}
    x_{\{A,B\}} & \leq s_A + s_B, \hspace{5mm} x_{\{A,B\}}  \geq s_A - s_B,\\
    x_{\{A,B\}} & \geq s_B - s_A, \hspace{5mm} x_{\{A,B\}} \leq 2 - s_A - s_B.
\end{aligned}
\end{equation*}

Following \Cref{thm:template-distinguishability}, we add the following constraint for every $t \in T_d$ to ensure that every attacker-distinguishable test case is contract distinguishable:
\begin{align}\label{eq:coverage}
    \sum_{A \in \mathit{SD}_t}{s_A} + \sum_{\{A,B\} \in \mathit{XOR}_t}{x_{\{A,B\}}} \geq 1.
\end{align}

Finally, to ensure that the $\textit{fp}_t$ variables are consistent with the contract distinguishability of the test cases we add the following constraints for each attacker-indistinguishable test case $t \in T_{nd}$: 
\begin{align*}
    \textit{fp}_t & \geq s_A \hspace{9.58mm}\text{ for every } A \in \mathit{SD}_t,\\
    \textit{fp}_t & \geq x_{\{A,B\}} \hspace{4.68mm}\text{ for every } \{A,B\} \in \mathit{XOR}_t.
\end{align*}

    \begin{restatable}{proposition}{ilpsynthesis}
    \label{thm:ilp-synthesis}
    Any solution to the ILP corresponds to  a set of atoms $S \subseteq \template$ such that $\contract_S$ distinguishes all attacker-distinguishable test cases and as few attacker-indistinguishable test cases as possible.
    \end{restatable}

\revision{
Mohr et al.~\cite{Mohr24} employ a similar ILP for contract synthesis.
As discussed in \Cref{sec:synthesis:testing}, their work lacks the notion of xor-distinguishing atoms, which is needed to support contracts where multiple atoms share the same leakage function. 
As a consequence, their ILP differs from ours in two ways:
\begin{inparaenum}[(a)]
    \item It omits all constraints involving the $x_{\{A,B\}}$ variables.
    \item To ensure that every attacker-distinguishable test case is contract distinguishable, it includes the following constraint for every $t \in T_d$: $\sum_{A \in \mathit{D}_t}{s_A} \geq 1$,
where $\mathit{D}_t$ is the set of distinguishing atoms for test case $t$.
\end{inparaenum}
}

\subsection{Phase 2b: Bounded soundness verification}\label{sec:synthesis:bounded-verification}

The ILP-based synthesis from \Cref{sec:synthesis:synthesis} generates a new contract based on the attacker and template distinguishability derived from the test cases (\Cref{sec:synthesis:testing}).
Since the test cases may miss some leaks in the target processor, the contract generated from \emph{Phase 2a} may be unsound.
To discover leaks that are missed by the candidate contract, \tool{} performs a bounded verification step.
\revision{This is another key difference between \tool{} and the approach by Mohr et al.~\cite{Mohr24} (beyond those outlined in \Cref{sec:synthesis:testing}--\ref{sec:synthesis:synthesis}): while the latter simply synthesizes the most precise contract from an initial set of test cases, without any soundness guarantees, \tool{} employs bounded verification to discover missed leaks, as we describe next, and provides soundness guarantees (as we show in \Cref{sec:synthesis:unbounded-verification}).}

\mypara{Intuition}
\tool{} verifies whether the candidate contract $\contract_S$ captures all leaks visible by the attacker $\atk$ on the target CPU $\uarch$ when considering all possible executions up to a given length $k$.
That is, \tool{} verifies a bounded variant of \Cref{def:uarch-contract-satisfaction}.
For this, \tool{} encodes the contract verification task as a bounded model checking (BMC) problem.
Since microarchitectural contract satisfaction is a property defined over pairs of $\uarch$ executions, our encoding relies on self-composition~\cite{BartheDR11}, which reduces reasoning about pairs of executions to reasoning about a single execution, by first constructing a product circuit consisting of two copies of $\uarch$.
The bounded contract satisfaction property $\phi_{\mathit{ctrsat}}$ is then encoded on top of this  circuit.
The BMC either (a) falsifies $\phi_{\mathit{ctrsat}}$ and provides a counterexample---a test case with the same contract traces but different attacker traces---which can be used to synthesize a better contract, or (b) proves that the contract is sound up to bound $k$.
 
\mypara{Constructing the product circuit}
We construct a product circuit consisting of two copies of $\uarch$ executing in parallel.
To compare contract observations, which are produced only when instructions retire (\ie whenever the user-provided retirement predicate $\psi$ holds), we need to synchronize the two copies of $\uarch$ on retirement.
That is, whenever one copy of $\uarch$ retires an instruction but the other does not (\ie $\psi$ holds only on one of the two executions), the copy of $\uarch$ where $\psi$ holds is ``paused'' until the other one catches up.
To do so, \tool{} uses the stuttering product circuit construction introduced by Wang et al.~\cite{Wang23}, which constructs a product circuit synchronized on the retirement predicate $\psi$.

\mypara{Property}
We encode microarchitectural contract satisfaction on top of the product circuit into the $\phi_{\mathit{ctrsat}}^{C,\psi}$ property, where $C$ is the contract and $\psi$ is the retirement predicate, as follows.
Our encoding is parametric in: (a) the total bound $k$, and (b) the attacker bound $1 \leq b \leq k$.
Intuitively, the formula  $\phi_{\mathit{ctrsat}}^{C,\psi}$ works as follows:
\begin{asparaenum}
\item it assumes that, at cycle $0$, the two copies of $\uarch$ start from a valid initial state and with the same microarchitectural part,
\item it assumes that, for cycles $0$ to $k$, the two executions produce the same contract traces, \ie  whenever the retirement predicate $\psi$ holds on both executions, the contract observations are the same, 
\item it asserts that, for cycles $0$ to $b$, the two executions have the same attacker observations, \ie they are attacker-indistinguishable.
\end{asparaenum}

Note that the attacker bound $b$ is less than the total bound $k$, since contract observations are ``slower'' than attacker observations, as the former are produced only when instructions retire.
This requires looking ahead to check whether future contract observations may ``declassify'' a difference in attacker observations.
Since these differences are often caused by in-flight instructions, $k - b$ needs to be large enough to account for the retirement of these instructions to ensure that the associated observations are produced.
However, sizing $k-b$ to the worst-case results in large bounds that (a) slow down checking contract satisfaction in simpler cases, and (b) may result in  counterexamples with many instructions.\looseness=-1

To account for this, \tool{} uses a simple optimization, which is parametric in an instruction bound $i$.
As soon as a difference in attacker observations is detected (in the first $b$ cycles), we start decrementing $i$ whenever an instruction is retired by both executions.
\tool{} then terminates the BMC as soon as either $i$ reaches $0$ or the cycle bound $k$ is reached, whichever comes first.
This optimization allows us to dynamically adjust the number of cycles explored by the BMC.
As we show in \Cref{sec:evaluation:results}, this has two benefits: (a) it speeds up bounded verification since the BMC finds many simple counterexamples faster, without having to always explore all $k$ cycles, and (b) it often results in shorter counterexamples with fewer instructions, which restrict the synthesizer's search space.
The encoding of the property checked by \tool{} in terms of \texttt{assume} and \texttt{assert} statements in Verilog is given in \techReportAppendix{appendix:property}.

\revision{
    \begin{restatable}{proposition}{boundedverificationcorrectnessa}
    \label{prop:bounded-verification:correctness-1}
    Let $\mathbb{K}$ be the maximum number of cycles that the CPU under verification takes to retire an instruction.
    If the BMC falsifies $\phi^{C,\psi}_{\mathit{ctrsat}}$, with retirement predicate $\psi$, BMC bound $k$, attacker bound $1\leq b \leq k$, and instruction bound $i$, then there is 
    an attacker-distinguishable test case $T$ 
    where the prefixes of length  $\mathit{min}(\left\lfloor \frac{k}{\mathbb{K}}\right\rfloor,  i)$ of the $C$-traces are identical.
    \end{restatable}

   From a falsification of $\phi^{C,\psi}_{\mathit{ctrsat}}$ we extract a synthetic attacker-distinguishable test case, which is then used to synthesize the next contract candidate via the ILP introduced in \Cref{sec:synthesis:synthesis}.
   To this end, we determine the test case's set of strongly-distinguishing atoms and xor-distinguishing pairs in the prefix of length $\mathit{min}(\left\lfloor \frac{k}{\mathbb{K}}\right\rfloor, i)$.
   This process is repeated until the BMC does not discover any more unaccounted leaks.
   Each iteration of the loop in Phase~2 rules out at least one of the template's contracts, and thus eventually the BMC proves $\phi^{C,\psi}_{\mathit{ctrsat}}$ and the synthesis loop terminates.
}

\revision{
    \begin{restatable}{proposition}{boundedverificationcorrectnessb}
        \label{prop:bounded-verification:correctness-2}
    Let $\mathbb{K}$ be the maximum number of cycles that the CPU under verification takes to retire an instruction.
    If the BMC proves $\phi^{C,\psi}_{\mathit{ctrsat}}$, with retirement predicate $\psi$, BMC bound~$k$, attacker bound $1\leq b \leq k$, and instruction bound $i$, then any $C$-indistinguishable test case  $T$ consisting of at most $\mathit{min}(\left\lfloor \frac{k}{\mathbb{K}}\right\rfloor,  i)$ instructions is also attacker indistinguishable for the first $b$ cycles.
    \end{restatable}
    \Cref{prop:bounded-verification:correctness-2} formalizes what we refer to as a boundedly-sound contract in the overview in \Cref{fig:approach-overview}.
}

\subsection{Phase 3: Unbounded soundness verification}\label{sec:synthesis:unbounded-verification}

As soon as {Phase 2b} cannot discover further counterexamples, \tool{} performs an {unbounded verification} step to ensure that the synthesized contract is \emph{sound}, \ie all possible leaks in the target processor are captured by the contract.

For this, \tool{} uses the \textsc{LeaVe} contract verifier~\cite{Wang23} to prove that (unbounded) microarchitectural contract satisfaction holds for the current contract and the target CPU.
In a nutshell, \textsc{LeaVe} verifies unbounded contract satisfaction by (a) learning the strongest possible inductive invariant over pairs of contract-in\-dis\-tin\-guish\-a\-ble executions, and (b) use this invariant to prove that contract-in\-dis\-tin\-guish\-a\-ble executions are also attacker-in\-dis\-tin\-guish\-a\-ble.

If \textsc{LeaVe}'s verification passes, the synthesized contract is sound  (from \cite[Theorem 2]{Wang23}) and as precise as possible (from \Cref{thm:ilp-synthesis}) for the target processor.
If it fails, \tool{} terminates 
 and notifies the user.
This failure can be due to two reasons:
\begin{asparaitem} 
    \item If the contract is unsound, the user can re-run \tool{} with a larger bound for {Phase 2b} to identify the missed leaks.\looseness=-1
    \item If the contract is sound but \textsc{LeaVe} fails to verify it, the user can provide additional information to \textsc{LeaVe} (e.g., additional relational invariants) to make the unbounded proof go through~\cite{Wang23}.
\end{asparaitem}

\begin{restatable}{proposition}{overallcorrectness}
    \label{prop:phase3:correctness}
    If Phase 3 passes, the contract synthesized by \tool{} is {sound} and distinguishes as few attacker-distinguishable test cases as possible among the test cases explored during the synthesis process.
\end{restatable}

\section{Implementation}\label{sec:implementation}

In this section, we present \tool{}, a prototype that implements our synthesis approach from \Cref{sec:synthesis} for CPU designs in Verilog.
Our prototype relies on the Yosys Open Synthesis Suite~\cite{yosys} for processing Verilog circuits, the Icarus Verilog simulator~\cite{iverilog} for simulating test cases, the Google OR-Tools~\cite{googleILP} for the ILP, the Yices SMT solver~\cite{yices} for the bounded verification, and the \textsc{LeaVe} contract verification tool~\cite{Wang23} for the unbounded verification.
\tool{}'s implementation is open-source and available at~\cite{artifact}.

\mypara{Inputs}
\tool{} takes as input 
\begin{inparaenum}[(1)]
\item the target processor $\uarch$ implemented in Verilog,
\item a Verilog expression $e$ over $\uarch$  expressing the retirement predicate,
\item a Verilog expression $\atk$ over $\uarch$ modeling the microarchitectural attacker,
\item the bounds for the bounded verification, and
\item the additional inputs for the unbounded verification in \textsc{LeaVe} (e.g., relational invariants).
\end{inparaenum}
The user can also provide expressions for initializing parts of $\uarch$'s state.\looseness=-1 %

\mypara{Contract templates}
\tool{} implements multiple contract templates (described in \Cref{sec:evaluation:methodology}) that capture different classes of instruction-level leaks.
Even though these templates are defined in terms of the RISC-V ISA, they need to be instantiated for each new target processor since different processors often implement the same architectural elements (e.g., the register file) with different signals in Verilog.
To address this, all our templates and atoms (across the entire \tool{}'s pipeline) are implemented in terms of the RISC-V formal interface~\cite{rvfi} (RVFI), which provides a common interface to the architectural state of a RISC-V core for testing and verification.
This interface, which is implemented by several open-source cores,  significantly simplifies applying \tool{} to a new CPU that already implements the RVFI.

\mypara{Test generation}
\tool{} can accept user-provided test cases as input.
Additionally, \tool{} implements a test generation strategy for automatically synthesizing a given number of test cases. 
For each leakage function in the implemented templates, we also implement the associated modifier functions, which \tool{} uses to generate test cases.
\tool{}'s test generation strategy follows the process described in \Cref{sec:synthesis:testing}: (a) we first randomly generate an architectural state $\sigma$ (consisting of a RISC-V program and initial values for registers), (b) we identify one of the instructions in the program, and (c) we use one of the applicable modifier functions for the selected instruction to generate the other initial state $\sigma'$.

\mypara{Workflow}
\tool{} implements the workflow described in \Cref{sec:synthesis}.
First, it generates a set of test cases, simulates them, and compute their attacker and template-distinguishability w.r.t. the attacker and selected template.
Then, it alternates between synthesis and bounded verification to come up with precise candidate contracts. %
Finally, it verifies whether the candidate contract is sound with the \textsc{LeaVe} (unbounded) contract verifier and terminates.

\newcommand{\DarkRiscv}{\textbf{DarkRISCV}}
\newcommand{\DarkRiscvTwoStages}{\DarkRiscv{}\textbf{-2}}
\newcommand{\DarkRiscvThreeStages}{\DarkRiscv{}\textbf{-3}}
\newcommand{\Sodor}{\textbf{Sodor}}
\newcommand{\SodorTwoStages}{\Sodor{}\textbf{-2}}
\newcommand{\SodorFiveStages}{\Sodor{}\textbf{-5}}
\newcommand{\SodorOneStage}{\Sodor\textbf{-1}}
\newcommand{\Ibex}{\textbf{Ibex}}
\newcommand{\IbexSmall}{\Ibex{}\textbf{-small}}
\newcommand{\IbexMul}{\Ibex{}\textbf{-mult-div}}
\newcommand{\IbexCache}{\Ibex{}\textbf{-cache}}
\newcommand{\BaseCtrTemplate}{\textbf{B}}
\newcommand{\AlignedCtrTemplate}{\textbf{A}}
\newcommand{\BranchCtrTemplate}{\textbf{BT}}
\newcommand{\ValueCtrTemplate}{\textbf{V}}
\newcommand{\DependenciesCtrTemplate}{\textbf{D}}
\newcommand{\InstructionCtrTemplate}{\textbf{I}}
\newcommand{\RegisterCtrTemplate}{\textbf{R}}
\newcommand{\MemoryCtrTemplate}{\textbf{M}}

\section{Evaluation}\label{sec:evaluation}

In this section, we report on the use of \tool{} to synthesize sound and precise leakage contracts.
We first introduce the methodology we followed for our evaluation (\Cref{sec:evaluation:methodology}): the processors we analyze, the contract templates and attacker we consider, and the experimental setup.
In our evaluation (\Cref{sec:evaluation:results}), we address these research questions:
\begin{asparaitem}
\item{\textbf{RQ1:}} Can \tool{} synthesize sound and precise leakage contracts from processor designs at RTL?

\item{\textbf{RQ2:}} What is the impact of the number of initial test cases on the synthesized contract's precision and on the synthesis time?

\item{\textbf{RQ3:}} What is the impact of different contract templates on the synthesized contract's precision?

\item{\textbf{RQ4:}} Does \tool{} synthesize contracts that are more precise than leakage contracts derived in prior work?

\item{\textbf{RQ5:}} What is the impact of different property encodings on the verification time of Phase 2b (\Cref{sec:synthesis:bounded-verification})?
\end{asparaitem}

We remark that all  benchmarks and scripts for reproducing our results are available at~\cite{artifact}.

\subsection{Methodology}\label{sec:evaluation:methodology}

\mypara{Benchmarks}
We analyze the following CPUs:
\begin{asparaitem}

    \item \DarkRiscv: An in-order RISC-V core that implements the RV32I instruction sets~\cite{darkriscv}.
    We analyzed the 2-stage \DarkRiscvTwoStages{} and the 3-stage \DarkRiscvThreeStages{} implementations. %
    
    \item \SodorTwoStages{}: An educational 2-stage RISC-V core~\cite{sodor} that implements the RV32I instruction set.\footnote{We translate Sodor (written in Chisel) into Verilog for the analysis with \tool{}.\looseness=-1}

    \item \Ibex{}: An open-source, production-quality 32-bit RISC-V CPU core~\cite{ibex}.\footnote{We translate Ibex (written in SystemVerilog) into Verilog for the analysis with \tool{}.\looseness=-1} %
    We target the default ``small'' configuration, %
    which 
    has two stages and supports the RV32IM instruction set.
    We study three variants:
    (1) \IbexSmall{} is the default ``small'' configuration with constant-time multiplication (three cycles) and without caches,
    (2) \IbexCache{} is \IbexSmall{} extended with a simple (single-line) cache, and
    (3) \IbexMul{} employs a non-constant-time multiplication unit whose execution time depends on the operands~\cite{multibex}.
    For all variants, compressed instructions are disabled.

\end{asparaitem}

\mypara{Contract templates}
We consider the following basic templates:
\begin{asparaitem}
\item \textbf{Instruction Leaks} (\InstructionCtrTemplate{}): atoms exposing  information about an instruction's encoding  (op code, destination and source registers, immediate values).
\item \textbf{Register Leaks} (\RegisterCtrTemplate{}): atoms exposing the value of an instruction's source and destination registers, and the program counter.
\item \textbf{Memory Leaks} (\MemoryCtrTemplate{}): atoms exposing the accessed memory address and content read from/written to  memory.
\item \textbf{Alignment Leaks} (\AlignedCtrTemplate{}):  two atoms that expose the alignment of a memory access: $IS\_ALIGNED$ exposes whether the last two bits of the memory address are $00$ and $IS\_HALF\_ALIGNED$ exposes whether the last two bits of the memory address are not $11$.\looseness=-1
\item \textbf{Branch Leaks} (\BranchCtrTemplate{}):  atoms that expose the outcome of branch instructions. For conditional jumps, the atoms expose if the branch is taken. For unconditional ones, the atoms return the constant $1$. %
\item \textbf{Value Leaks} (\ValueCtrTemplate{}):  atoms exposing selected information about values of  source and destination registers: (a) whether these values are equal to $0$, and (b) the logarithm of the value. These atoms are inspired by known value-dependent leaks in the \Ibex{} core~\cite{Wang23,Mohr24}.
\end{asparaitem}

Additionally, we consider templates obtained by combining the basic ones listed above, and we denote such combinations with $+$.
In particular, the \textbf{Base Leaks} (\BaseCtrTemplate) template consists of $\InstructionCtrTemplate{} + \RegisterCtrTemplate{} + \MemoryCtrTemplate{}$, that is, the template containing all atoms from $\InstructionCtrTemplate{}$,  $\RegisterCtrTemplate{}$, and $\MemoryCtrTemplate{}$.

\mypara{Attacker} 
We consider an attacker that observes when instructions retire, \ie 
 when the retirement predicate holds.

\mypara{Unbounded verification setup}
For the unbounded verification using \textsc{LeaVe}, for all benchmarks, we re-use the initial-state invariants and the additional candidate inductive invariants manually constructed by Wang et al.~\cite{Wang23} for verifying the same cores.

\mypara{Experimental setup} 
All our experiments ran on an AMD Ryzen Threadripper PRO 5995WX with 64 cores and 512 GB of RAM. Each experiment ran in a dockerized Ubuntu 24.04 with OpenJDK 21. 
\tool{} used Google OR-Tools version 9.10, Yosys version 0.48, Icarus Verilog version s20250103, and Yices version 2.6.5.

\mypara{Evaluating precision}
Leakage contracts classify pairs of executions as distinguishable or indistinguishable.
To measure the precision of synthesized contracts, we use the standard notion of precision used for binary classifiers, following~\cite{Mohr24}.
Given a target processor, a contract $C$, and a validation set $V$ of test cases, the precision of $C$ with respect to $V$ is $\frac{\mathit{FP}}{\mathit{TP}+\mathit{FP}}$, where $\mathit{TP}$ is the number of test cases that are contract and attacker distinguishable and $\mathit{FP}$ is the number of test cases that are contract distinguishable but attacker indistinguishable.
For each benchmark, we generate a validation set consisting of $2048K$ test cases.

\subsection{Experimental results}\label{sec:evaluation:results}

\mypara{RQ1 -- Synthesis for open-source RISC-V cores}
To evaluate whether \tool{} can  synthesize sound and precise leakage contracts from open-source cores, we apply \tool{} to all benchmarks from \Cref{sec:evaluation:methodology}. %
For each core, we use \tool{} to synthesize a contract against the $\BaseCtrTemplate + \AlignedCtrTemplate + \BranchCtrTemplate + \ValueCtrTemplate$ template
with a set of $128K$ test cases.\looseness=-1

\begin{table*}[h]

    \caption{Results of \tool{}'s synthesis campaign.
}\label{tab:verification:results}
\vspace{-10pt}
    \centering
    \small
\begin{tabular}{l c c c c c c c c c c c }
\toprule
    \multirow{2}{*}{\textit{Processor}} & \multirow{2}{*}{\textit{Iterations}} & \textit{BMC} & \textit{ATK} & \textit{Instr.}   & \multicolumn{4}{c}{\textit{Synthesis time} (hh:mm:ss)} & \multirow{2}{*}{\textit{\# atoms}} & \multirow{2}{*}{\textit{Precision}}  \\
     &  & \textit{bound}& \textit{bound}  & \textit{bound} & \textit{Phase 1} & \textit{Phase 2a} & \textit{Phase 2b} & \textit{Phase 3} & &&     \\
    \midrule
    \DarkRiscvTwoStages{}   & 1     & 15    & 10    & 1 & 00:02:41  & 00:00:17 &	01:22:44 & 43:16:44 &  8    & 0.998\\ %
    \DarkRiscvThreeStages{} & 1     & 20    & 15    & 1 & 00:03:14	& 00:00:21 &	04:52:56 & 18:51:26 &  8    & 0.999 \\ %
    \SodorTwoStages{}       & 1     & 6     & 5     & 1 & 00:04:38  & 00:00:09 &	00:13:31 & 03:01:59 &  8    & 0.998 \\
    \IbexSmall{}            & 4     & 51    & 12    & 1 & 00:07:08  & 00:02:09 &	02:38:33 & 22:09:32 &  24   & 0.999 \\ %
    \IbexCache{}            & 24    & 51    & 12    & 1	& 00:07:18  & 00:14:45 &	02:11:28 & 21:43:59 &  27   & 0.579 \\ %
    \IbexMul{}              & 2     & 51    & 12    & 1	& 00:07:32  & 00:01:12 &	03:28:17 & 24:38:21 &  24   & 0.999 \\ %
    \bottomrule
\end{tabular}
\vspace{-5pt}
\end{table*}

\Cref{tab:verification:results} summarizes the results of our evaluation. %
The table reports 
\begin{inparaenum}[(a)]
\item the number of iterations for the synthesis loop,
\item the static BMC bound~$k$, the attacker bound~$b$ and the instruction bound~$i$ used, 
\item the total synthesis time split between Phase 1, Phase 2a, Phase 2b, and Phase 3,
\item the number of atoms in the final synthesized contract, and
\item the precision of the synthesized contract.
\end{inparaenum}

We highlight the following findings:
\begin{asparaitem}
\item \tool{} successfully synthesizes sound contracts for all target CPUs, and all contracts pass the unbounded verification phase.

\item The synthesized contracts have precision higher than $0.99$ for all targets except for \IbexCache.
The low precision in the \IbexCache{} case is due to the core having a single-line cache, which leaks whether a given address has been accesses consecutively or not.
This leak cannot be precisely captured in the $\BaseCtrTemplate + \AlignedCtrTemplate + \BranchCtrTemplate + \ValueCtrTemplate$ template, and \tool{} falls back to exposing the values of all memory accesses.

\item For all cores, \tool{} can synthesize a contract that passes the bounded verification (without Phase 3) in less than 5 hours.
In several cases, the unbounded verification of the contract using \textsc{LeaVe}~\cite{Wang23} dominates the overall execution time.

\item For \DarkRiscvTwoStages, \DarkRiscvThreeStages, and \SodorTwoStages, \tool{} can synthesize a sound contract (each one with $8$ atoms) directly from the initial test cases.
For \Ibex{}, instead, \tool{} needs to iterate between Phases 2a and 2b, before converging on the final contract.
Furthermore, the synthesized contracts contain more atoms than those for \DarkRiscv{} and \SodorTwoStages{}, which is consistent with \Ibex{} being more complex than the other benchmarks.

\item 
While \SodorTwoStages, \DarkRiscvTwoStages, and \DarkRiscvThreeStages{} only leak through control-flow operations, contracts for the \Ibex{} cores capture additional leaks.
The contracts for all \Ibex{} cores expose:
\begin{inparaenum}[(a)] 
    \item if the second operand of \texttt{div}, \texttt{divu}, \texttt{rem}, \texttt{remu} instructions is $0$,
    \item if memory accesses are aligned or not (only for stores for \IbexCache{}),
    \item if the core is executing a memory read or write, and
    \item if the core is executing a \texttt{mul}, \texttt{mulh}, \texttt{mulhu}, or \texttt{mulhsu}.
\end{inparaenum}
Furthermore, the contract for \IbexCache{} exposes the entire address of memory loads (capturing the cache leak), whereas the contract for \IbexMul\ exposes the logarithm of the second operand for \texttt{mul} instructions (capturing the leak introduced by the slow-multiplication unit).\looseness=-1
\onlyTechReport{\Cref{tab:contract:leaks} in \Cref{appendix:data} summarizes these leaks.}
\end{asparaitem}

\mypara{RQ2 -- Impact of the number of  test cases}
\tool{} uses test cases to generate the candidate contract.
To evaluate the impact of the number of test cases on the synthesis process, we use \tool{} to synthesize a contract for all our benchmarks using the  $\BaseCtrTemplate + \AlignedCtrTemplate + \BranchCtrTemplate + \ValueCtrTemplate$ %
 template and different number of test cases (from 0 to $2048K$).

\Cref{fig:eval:rq2} reports the results of our experiments.
In particular, \Cref{fig:eva-time} reports the execution time (excluding the unbounded verification)\footnote{
We omit the unbounded verification time because (a) it is roughly independent on  the number of test cases and (b) it often dominates the rest of the process.} needed to synthesize a sound contract for the different configurations, \Cref{fig:eva-testcase} reports the precision of the synthesized contract, \revision{and \Cref{fig:eva-time-share} reports the fraction of time spent on synthesizing the first contract candidate, which includes all of Phase~1 and the first call to the ILP solver (Phase~2a), relative to the total execution time (excluding unbounded verification).}

We highlight the following findings:
\begin{asparaitem}
\item A very small number of initial test cases often results in a high synthesis time.
The reason for this is that small sets of initial test cases often exercise only a limited behavior of the target processor, which results in a greater number of iterations of the synthesis loop.\looseness=-1 %
\item Increasing the number of test cases is beneficial as long as the simulation time required to find a new leak is shorter than the time required to find a new counterexample via bounded verification.
At around 1024k/2048k test cases the probability of finding new leaks via simulation is low and the simulation time starts to dominate.

\item A very small number of test cases often results in low precision contracts.
This is because Phase 2a (\Cref{sec:synthesis:synthesis}) relies on attacker-indistinguishable test cases to optimize the precision of the synthesized contract.
Small sets of test cases do not provide enough information to guide synthesis and often leads to imprecise contracts. 

\item Increasing the number of test cases usually increases the precision of the synthesized contract, even though there are cases where the  contract gets worse with more test cases.
For all cases except $\IbexCache$, the precision of the synthesized contract stabilizes at 1 quickly (at most after reaching $32K$ test cases).
For $\IbexCache$, the precision stabilizes at $0.579$.
This is due to the limitations of the template.
All synthesized contracts for $\IbexCache$ expose the accessed memory address (to capture cache leaks), even though $\IbexCache$ implements a single-line cache so the actual leak only exposes whether a given address has been accessed consecutively or not, which cannot be captured in our template.
\end{asparaitem}

\begin{figure}
    \centering
    \begin{subfigure}[b]{0.49\textwidth}
        \centering
        \input{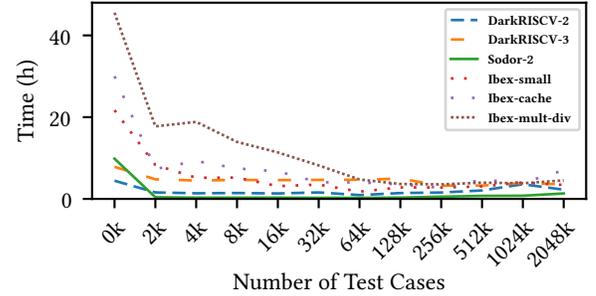}
        \vspace{-10pt}
        \caption{Synthesis time (excluding unbounded verification)} %
        \label{fig:eva-time}
    \end{subfigure}
    \hfill
    \begin{subfigure}[b]{0.49\textwidth}
        \centering
        \input{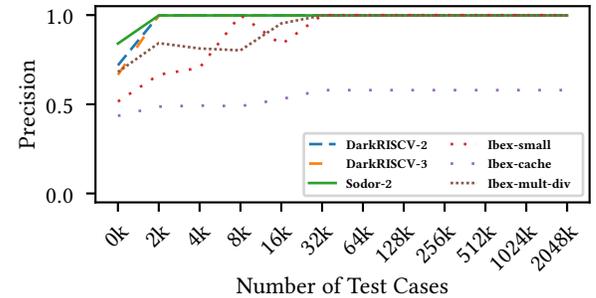}
        \vspace{-10pt}
        \caption{Precision of the synthesized contract} %
        \label{fig:eva-testcase}
    \end{subfigure}
    \revision
{    \begin{subfigure}[b]{0.49\textwidth}
        \centering
        \input{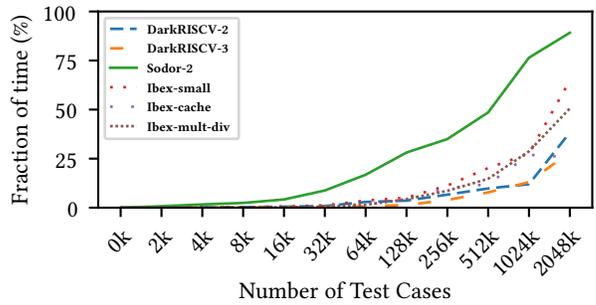}
        \vspace{-10pt}
        \caption{Fraction of time spent on synthesizing the first contract candidate (\ie Phase 1 and one iteration of Phase 2a) relative to the total synthesis time (excluding unbounded verification)
        }
        \label{fig:eva-time-share}
    \end{subfigure}}
    \vspace{-15pt}
    \caption{Impact of the number of test cases on contract synthesis}      
    \label{fig:eval:rq2}
    \vspace{-10pt}
\end{figure}

\mypara{RQ3 -- Impact of contract templates}
The user-provided contract template defines the search space explored by \tool{}.
To study the impact of different templates, we use \tool{} to synthesize a contract with increasingly more expressive templates and measure the contract's precision.
For each core and  template, we run \tool{} with the same bounds as in \Cref{tab:verification:results} and with $64K$ test cases.\looseness=-1

\Cref{fig:eva-ctr-families} summarizes our results. We highlight these findings:
\begin{asparaitem}

    \item Template $\BaseCtrTemplate = \InstructionCtrTemplate + \RegisterCtrTemplate + \MemoryCtrTemplate$ is the first template for which \tool{} can synthesize sound contracts for all  target CPUs.
    For templates $\InstructionCtrTemplate$, $\RegisterCtrTemplate$, $\MemoryCtrTemplate$ in isolation, \tool{}  cannot synthesize a sound contract, since no sound contract exists in these templates. %

    \item $\BaseCtrTemplate \to \BaseCtrTemplate + \AlignedCtrTemplate:$ For the \Ibex{} cores, additionally considering template \AlignedCtrTemplate{} results in more precise contracts.
    Executing aligned and non-aligned memory accesses take different times.
    In template~$\BaseCtrTemplate$, this can only be captured by exposing the entire memory address.
    In~$\BaseCtrTemplate + \AlignedCtrTemplate$, \tool{} synthesizes a more precise contract that exposes only whether memory accesses are aligned or not.

    \item $\BaseCtrTemplate + \AlignedCtrTemplate \to \BaseCtrTemplate + \AlignedCtrTemplate + \BranchCtrTemplate:$ 
    Additionally considering the $\BranchCtrTemplate$ template increases precision.
    The branch-taken and branch-non-taken cases have different execution times.
    Without $\BranchCtrTemplate$, \tool{} is forced to synthesize contracts that expose the values of the registers used to compute the branch. 
    With $\BranchCtrTemplate$, \tool{}  synthesizes more precise contracts that only expose whether the branch is taken.

    \item $\BaseCtrTemplate + \AlignedCtrTemplate + \BranchCtrTemplate \to \BaseCtrTemplate + \AlignedCtrTemplate + \BranchCtrTemplate + \ValueCtrTemplate:$ 
    For the \Ibex{} cores, additionally considering the template $\ValueCtrTemplate$ further increases precision.
    These cores rely on a multiplication/division unit where (a) the execution time of multiplications is proportional to the logarithm of one of its operands, and (b) there is a fast path for division if the divisor equals to $0$.
    With \ValueCtrTemplate, \tool{} can synthesize contracts that characterize such leaks, rather than having to expose entire operands.
\end{asparaitem}

\begin{figure}
    \centering
    \input{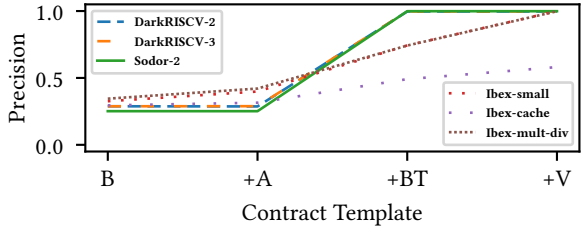}
    \vspace{-15pt}
    \caption{Impact of template expressiveness on precision} %
    \label{fig:eva-ctr-families}
    \vspace{-15pt}
\end{figure}

\mypara{RQ4 -- Comparison with contracts from prior work}
Prior work have derived leakage contracts for some \mbox{RISC-V} cores.
We focus on contracts from \textsc{LeaVe}~\cite{Wang23}, Mohr et al.~\cite{Mohr24}, ConjunCT~\cite{conjunct}, VeloCT~\cite{h-houdini}, and RTL2M$\mu{}$PATH~\cite{rtl2mupath} for the \IbexSmall{} CPU and we compare them (in terms of precision) to the contract synthesized by \tool{} using $\BaseCtrTemplate + \AlignedCtrTemplate + \BranchCtrTemplate + \ValueCtrTemplate$ as template, an initial set of $64K$ test cases, and the bounds from~\Cref{tab:verification:results}. 
Like \tool{}, \cite{Wang23,Mohr24,conjunct,h-houdini,rtl2mupath} target an attacker that observes retirement time.
\revision{
Both \textsc{LeaVe}~\cite{Wang23} and  ConjunCT~\cite{conjunct} have been applied to \IbexSmall{}, and we consider the contracts reported in the corresponding papers.
For Mohr et al.~\cite{Mohr24}, we consider a contract synthesized with $64K$ initial test cases.
VeloCT~\cite{h-houdini} (which is an improvement over ConjunCT) lacks an evaluation against \IbexSmall\ (and it has only been applied to Boom in~\cite{h-houdini}); we contacted the paper's authors who confirmed that VeloCT produces the same contract as ConjunCT for Ibex, therefore for VeloCT we use the same contract as for ConjunCT.
Finally, for RTL2M$\mu{}$PATH~\cite{rtl2mupath} (which also lacks an evaluation against \IbexSmall{} and is built on top of the JasperGold commercial verification tool),  we consider the most precise contract synthesized by \tool{} for \IbexSmall{} against a template capturing RTL2M$\mu{}$PATH's contracts, \ie a template that includes atoms that leak the operands of instructions and forces a leakage of the instruction word and the program counter.
}

We measure the precision of the four contracts against two different validation set: \textbf{different-programs} and \textbf{same-programs}.
The \textbf{different-programs} validation set consists of 2048K test cases where the two programs in each test case might be different, which we generate with \tool{}'s default test-generation strategy.
The \textbf{same-programs} validation set consists of 2048K test cases where the two programs in each test case are the same and only the initial inputs differ.
We consider also \textbf{same-programs} since~\cite{Wang23,conjunct,h-houdini,rtl2mupath} assume that the program itself is public and known to the attacker.

\revision{\Cref{fig:eva-comparison} summarizes the precision of all contracts.
The contract generated by \tool{} is significantly more precise than those from~\cite{Wang23,conjunct,h-houdini,rtl2mupath} when considering the \textbf{different-programs} validation set.
This is unsurprising as the contract templates underlying~\cite{Wang23,conjunct,h-houdini,rtl2mupath} treat as contract distinguishable any two executions involving different programs/instruction streams.
However, the contract produced by \tool{} is also more precise than the ones from~\cite{Wang23,conjunct,h-houdini,rtl2mupath} under the \textbf{same-programs} validation set.
This is due to the fact that~\cite{Wang23,conjunct,h-houdini,rtl2mupath} consider coarser templates than $\BaseCtrTemplate + \AlignedCtrTemplate + \BranchCtrTemplate + \ValueCtrTemplate$, which require them to expose more information than needed to capture the actual leaks in \IbexSmall.}

\revision{
Finally, we compare the sound contracts from \tool{} with the (potentially unsound) ones produced by the approach from Mohr et al.~\cite{Mohr24}.
For this, we synthesize contracts using  \tool{} and the approach from~\cite{Mohr24} for  \IbexSmall\ with different numbers of initial test cases (from 1K to 2048K).
For all synthesized contracts, we (a) measure precision against the \textbf{different-programs} validation set, and (b) for the contracts from \cite{Mohr24} we also checked their soundness.

\Cref{fig:eva-date} reports the soundness/precision of all synthesized contracts.
The contracts produced by \tool{} are more precise (except for those produced with 1K and 2K initial test cases) than those produced by Mohr et al.~\cite{Mohr24}.
In particular, with the maximum number of test cases (2048k), Mohr et al.'s approach reaches a precision of 0.69 whereas \tool{}'s contract attains a precision of 0.999. 
None of the contracts generated by the approach from Mohr et al.~\cite{Mohr24} is sound whereas, as expected, all contracts from \tool{} are sound. 
}

\begin{table*}[h]
\begin{threeparttable}
    \caption{Comparison of precision with contracts from previous works~\cite{Wang23,conjunct,h-houdini,rtl2mupath,Mohr24}}\label{fig:eva-comparison}
    \centering
    \small
\begin{tabular}{l c c c c c c }
\toprule
    & LeaVe~\cite{Wang23} & \revision{Mohr et al.~\cite{Mohr24}}\tnote{$\ddagger$} & ConjunCT~\cite{conjunct} & \revision{VeloCT~\cite{h-houdini}} & \revision{RTL2M$\mu{}$PATH~\cite{rtl2mupath}} & Our work\\
    \midrule
    \textbf{same-programs} & 0.774 & \revision{0.855} &0.544 & \revision{0.544} & \revision{0.616} &1.000\\
    \textbf{different-programs} &0.054 & \revision{0.784} &0.056 & \revision{0.056} & \revision{0.064} &0.999 \\
    \bottomrule
\end{tabular}
\revision{
    \begin{tablenotes}
        \footnotesize
        \item[$\ddagger$] Unsound
    \end{tablenotes}
}
\end{threeparttable}
\end{table*}

\newcommand{\daggertwo}{\tnote{$\ddagger$}}

\begin{table*}[h]
    \begin{threeparttable}
    \caption{Comparison of precision and soundness with \cite{Mohr24}}\label{fig:eva-date}
    \centering
    \small
\begin{tabular}{l c c c c c c c c c c c c c}
\toprule
     & 1K &  2K & 4K & 8K & 16K & 32K & 64K & 128K & 256K & 512K & 1024K & 2048K\\
    \midrule
    Mohr et al.~\cite{Mohr24} & 0.706\daggertwo & 0.738\daggertwo & 0.697\daggertwo & 0.733\daggertwo & 0.750\daggertwo & 0.781\daggertwo & 0.784\daggertwo & 0.784\daggertwo & 0.784\daggertwo & 0.784\daggertwo & 0.699\daggertwo & 0.693\daggertwo \\
    Our work & 0.515\tnote{$\dagger$} & 0.667\tnote{$\dagger$} & 0.704\tnote{$\dagger$} & 0.999\tnote{$\dagger$} & 0.835\tnote{$\dagger$} & 0.999\tnote{$\dagger$} & 0.999\tnote{$\dagger$} & 0.999\tnote{$\dagger$} & 0.999\tnote{$\dagger$} & 0.999\tnote{$\dagger$} & 0.999\tnote{$\dagger$} & 0.999\tnote{$\dagger$} \\
    \bottomrule
\end{tabular}
    \begin{tablenotes}
        \footnotesize
        \item[$\dagger$] Sound \hspace{3mm} $\ddagger$ Unsound
    \end{tablenotes}
\end{threeparttable}
\vspace{-5pt}
\end{table*}

\mypara{RQ5 -- Impact of property encoding on bounded verification}
We use \tool{} to synthesize contracts for the \IbexSmall{} core with $\BaseCtrTemplate + \AlignedCtrTemplate + \BranchCtrTemplate + \ValueCtrTemplate$ as template.
We implement and compare two encodings:
the base one (denoted \textbf{b}) where the BMC always explores up to $51$ cycles (same bound as in \Cref{tab:verification:results}), and 
the optimized one (denoted \textbf{o} and described in \Cref{sec:synthesis:bounded-verification}) where the BMC stops the exploration as soon as: (a) one instruction retires after observing a difference on the attacker trace, or (b) it reaches $51$ cycles.

\Cref{fig:eva-encoding} summarizes the results. We highlight these findings:
\begin{asparaitem}
\item The optimized encoding results in faster (bounded) synthesis. 
However, the speedup is more significant with a lower number of initial test cases.
Part of this speedup is due to a reduced number of iterations along the synthesis loop.

\item The optimized encoding results in shorter counterexamples, \ie programs with fewer instructions.
Intuitively, shorter counterexamples are better since they reduce the synthesizer's search space.\looseness=-1
\end{asparaitem}

\newcommand{\savespace}[1]{\hspace{-6mm}#1}

\begin{table}
    \caption{ Impact of property encoding on bounded verification} \label{fig:eva-encoding}
    \centering
    \small
    \vspace{-7pt}
\begin{tabular}{l l r r r r r r r r r }
\toprule
    \multicolumn{2}{c}{\# test cases}                    & 2K        & 4K        & 8K      & 16K     & 32K     & 64K     & \savespace{128K}    & \savespace{256K}    \\
    \midrule
    Time    & \textbf{o}        & 8:18   & 5:10   & 5:12 & 3:07 & 3:30 & 1:46 & 2:48 & 2:47\\
    (hh:mm)                  & \textbf{b}    & \savespace{19:27}  & 9:37   & 7:39 & 6:19 & 3:58 & 2:32 & 1:41 & 2:43 \\
    \multirow{2}{*}{Iterations} & \textbf{o}             &  61       & 40        & 27      & 13      & 10      & 7       & 3       & 3              \\
                            & \textbf{b}                 &  110      & 48        & 37      & 33      & 15      & 9       & 2       & 3              \\
    avg. \# instr.\hspace{-5mm}              & \textbf{o}              &  1.1      & 1.3       & 1.1     & 1.2     & 1.2     & 1.3     & 1.7     & 1.7        \\
    in cex                  & \textbf{b}                 &  5.8      & 4.4       & 5.4     & 5.3     &  10.0     & 7.6       & 3.0       & 7.0        \\
    \bottomrule
\end{tabular}
\vspace{-5pt}
\end{table}

\section{Discussion}\label{sec:discussion}

\mypara{Verification bounds}
The choice of bounds for the bounded verification step (Phase 2b) can affect the success of the synthesis process and its runtime.
On the one hand, too short a bound can result in a candidate contract that passes the bounded verification but fails the final unbounded verification step (Phase 3), since the bounded verification failed in discovering some leaks.
On the other hand, too large a bound can significantly slow down the entire synthesis process, in particular when the synthesis loops requires many iterations.
In practice, in our experiments starting with a total bound of $i*k$ and an instruction bound of $i-1$, where $k$ is the maximum number of cycles for retiring an instruction and $i$ is the maximum number of in-flight instructions, was sufficient to successfully synthesize the contracts.

\mypara{Limitations on supported contracts}
Our formal model (\Cref{sec:formal-model}) is restricted to \emph{sequential} leakage contracts that only capture leaks generated by architectural instructions.
Using the terminology of~\cite{contracts2021,oleksenko2022revizor}, \tool{} only synthesizes ``leakage clauses'' for the sequential execution model.
Furthermore, given that the leakage functions in our atoms map single architectural states to contract observations, \tool{} is geared towards synthesizing contracts for instruction-level leaks, \ie those leaks that can be attributed directly to a specific instruction rather than  to a sequence of instructions.
We leave the synthesis of other classes of contracts, e.g., those capturing speculative leaks~\cite{contracts2021,speculation-at-fault}, as future work.

\mypara{Dependence on user-provided contract atoms}
Our approach critically relies on user-provided contract atoms that are able to capture the leakage of the target processor soundly and precisely.
For non-speculative leaks, it is fairly straightforward to define atoms able to soundly capture a processor's leakage, by exposing any information an instruction may depend on.
It is less obvious how to define atoms able to capture leakage precisely, as this requires understanding the actual leakage of the processor.
It is future work to investigate how to automatically derive such atoms.

\mypara{Termination of the synthesis loop}
The contract synthesis algorithm generates a sequence of candidate contracts in a non-monotonic fashion: each new candidate need not be a superset of the previous one.
For instance, $\ctr_2$ from \Cref{sec:overview:processor} is incomparable to $\ctr_1$.
Nevertheless, the synthesis loop is guaranteed to terminate, since each counterexample rules out at least one incorrect candidate, and the number of possible contracts given a template is finite.\looseness=-1

\mypara{(Non-)determinism of synthesis}
For a fixed set of test cases, there may be multiple incomparable sound contracts that have the same optimal precision.
Beyond this inherent non-determinism, however, the synthesis process is deterministic:
While the bounded verification step may generate different counterexamples in different runs, the algorithm will always converge to one of the sound contracts that are optimal with respect to the  test cases generated in Phase~1.\looseness=-1

\revision{\mypara{Leakage contracts and secure programming}
Leakage contracts provide the foundations for writing leak-free software.
As shown by Guarnieri et al.~\cite{contracts2021}, ensuring that program-level secrets do not influence contract traces is sufficient to ensure the absence of microarchitectural leaks for any CPU satisfying the contract.

In our framework, a contract $\contract_S$ is defined via a set $S$ of atoms $(\pi_A, \phi_A)$, where $\phi_A$ indicates what is exposed on the contract trace whenever $\pi_A$ holds.
Different atoms may share the same leakage function, \ie it is possible that $\phi_A = \phi_B$ for $A \neq B$, but different leakage functions must have disjoint images (\Cref{sec:contracts}). 
To guarantee leak freedom at program level, following~\cite{contracts2021}, one needs to ensure that the sequence of contract observations produced by the program is secret independent.
For each leakage function $\phi$, one thus needs to check that its applicability is secret independent, \ie that $\bigvee \{\pi_A \mid A \in S, \phi_A = \phi\}$ is secret independent, and whenever $\phi$ is applicable its value must be secret independent, too.

To check this non-interference property, one can adapt existing constant-time analysis tools~\cite{AlmeidaBBDE16,spectector2020,daniel2020binsec}, which are already able to account for various constant-time policies. 
A standard assumption, exploited by existing tools, is that the program path leaks and may thus not be secret dependent.
This is the case for all contracts synthesized in this work.
However, in principle, our framework allows for contracts that do not leak the outcome of branches, which has also been discussed in the context of efficient constant-time programming recently~\cite{Winderix24, Libra}
Adapting existing tools to handle such contracts in a precise manner may be more challenging.

}

\section{Related work}\label{sec:rw}
\mypara{Leakage verification}
LeaVe~\cite{Wang23} \emph{verifies} a given leakage contract against a RTL processor
 by learning a set of inductive relational invariants. 
Contract Shadow Logic~\cite{shadow-logic} follows a similar approach, but it uses an
explicit state model checker instead of relying on invariants.
While both approaches ensure soundness of contracts, these contracts are user-provided and their precision crucially relies on careful modelling by the user.
In contrast, our work \emph{synthesizes} sound and precise leakage contracts automatically, starting from a high-level contract template describing possible leakage sources. 

Given a hardware design, ConjunCT~\cite{conjunct} and its follow-up VeloCT~\cite{h-houdini} identify a set of
\emph{safe} instructions that can be invoked with secret arguments without timing leaks.
While the focus on identifying safe
instructions allows these tools to scale remarkably well, it also limits
the guarantees they provide. In particular, they only guarantee that programs consisting exclusively of safe instructions will not result in leaks.
However, since branching and memory instructions affect timing, they are not
safe, which severely limits set of programs that can be run securely on the
processor.

Ceesay-Seitz et al.~\cite{mu-cfi} introduce a property called microarchitectural control-flow
integrity ($\mu$-CFI), which checks for undue influences from instruction
operands to the (micro-architectural) program counter using a taint-tracking approach. 
While
$\mu$-CFI captures a number of security violations of processors (\eg certain
speculative execution attacks), it can only check for violations of a fixed
property and cannot express general leakage contracts.

Bloem et al.~\cite{power-contracts} verify leakage contracts for
power side channels. Their technique can verify the soundness of a given
contract, but it does not perform synthesis and offers no precision guarantees.

\mypara{Leakage Testing}
Revizor~\cite{oleksenko2022revizor} and Scam-V~\cite{nemati} check black-box CPUs
against leakage contracts via testing.
White box fuzzers like Phantom Trails~\cite{phantomtrails},
IntroSpectre~\cite{introspectre}, and SpecDoctor~\cite{specdoctor} can detect
speculative execution bugs in CPUs by fuzzing an RTL design, but they cannot
capture general contracts.
While these approaches can find violations for certain classes of leaks, they
cannot prove their absence and do not synthesize leakage descriptions.

\mypara{Synthesis of leakage descriptors}
RTL2muPath~\cite{rtl2mupath,rtl2muspec} synthesizes happens-before models for microarchitectural leaks.
It enumerates the
paths that an instruction can take through the microarchitecture. If different
paths have different timing profiles, the instruction may leak via timing.
The technique therefore records which other instructions may influence the choice of
microarchitectural path, thereby obtaining a form of leakage descriptor. 
While the leakage descriptors in \cite{rtl2mupath} are synthesized semi-automatically, unlike \tool, their synthesis is not guaranteed to be precise, and it is not obvious how to relate their descriptors back to software.  
In practice, the synthesized leakage descriptors may also be unsound, as the underlying model checker calls frequently time out~\cite{rtl2mupath}.

Mohr et al.~\cite{Mohr24} synthesize leakage contracts that are precise, but cannot guarantee their soundness differently from \tool{}.
Moreover, they cannot handle  contract templates where a leakage function is shared by multiple atoms, which \tool{} uses when synthesizing contracts for all our benchmarks in \Cref{sec:evaluation}. %

\section{Conclusion}\label{sec:conclusion}

To use leakage contracts to write software secure against microarchitectural attacks, programmers require sound and precise contracts for existing processors, which we currently lack.
To address this gap, we presented an approach to automatically derive such contracts directly from an RTL processor design, with limited user intervention.
We implemented our approach in the \tool{} synthesis tool, which we used to synthesize sound and precise contracts for six open-source RISC-V cores.
Our results indicate that contracts synthesized by \tool{} are more precise than sound contracts constructed in prior works.
\begin{acks}
We would like to thank Sushant Dinesh for the help in applying VeloCT~\cite{h-houdini} to Ibex.
This project has received funding from the \grantsponsor{1}{European Research Council}{https://erc.europa.eu/} under the European Union's Horizon 2020 research and innovation programme (grant agreements No. \grantnum{1}{101020415}, and \grantnum{1}{101115046}), 
from the \grantsponsor{2}{Spanish Ministry of Science and Innovation}{https://www.ciencia.gob.es/} under the project \grantnum{2}{TED2021-132464B-I00 PRODIGY}, 
from the \grantsponsor{3}{Spanish Ministry of Science and Innovation}{https://www.ciencia.gob.es/} under the Ram\'on y Cajal grant \grantnum{3}{RYC2021-032614-I}, 
from the \grantsponsor{3}{Spanish Ministry of Science and Innovation}{https://www.ciencia.gob.es/} under the project \grantnum{3}{PID2022-142290OB-I00 ESPADA}, 
and from the \grantsponsor{3}{Spanish Ministry of Science and Innovation}{https://www.ciencia.gob.es/} under the project \grantnum{3}{CEX2024-001471-M}.
\end{acks}

\bibliographystyle{ACM-Reference-Format}
\balance
\bibliography{biblio}

\onlyTechReport{
\clearpage
\nobalance
\appendix
\revision{
\section{Proofs}\label{appendix:proofs}

\subsection{Proof of \Cref{thm:template-distinguishability}}

\templatedistinguishability*

\begin{proof}
    Above \Cref{thm:template-distinguishability} in \Cref{sec:synthesis:template-distinguishability} we have argued that including a strongly-dis\-tin\-guish\-ing atom or exactly one of the two atoms in a xor-distinguishing pair is \emph{sufficient} to make a test case contract distinguishable.
    It remains to show that this is also \emph{necessary}.

    Let $t = (\sigma, \sigma')$ be a test case that is contract distinguishable under $\contract_S$.
    The test case must be contract distinguishable due to a \emph{leakage mismatch}, an \emph{applicability mismatch}, or both.

    \smallskip
    \noindent
    \emph{Case 1:} Assume there is leakage mismatch at index $i$.
    Thus, there are atoms $A,B \in S$ with $\phi_A = \phi_B$, $\pi_A(\sigma_i) \wedge \pi_B(\sigma_i)$, and $\phi_A(\sigma_i) \neq \phi_B(\sigma'_i)$.
    As the applicability predicates of atoms with the same leakage function are required to be mutually exclusive (see \Cref{sec:contracts}) no atoms other than $A$ and $B$ may be applicable in $\sigma_i$ and $\sigma'_i$, respectively.
    Thus, $A$ and $B$ are strongly distinguishing.

    \smallskip
    \noindent
    \emph{Case 2:} Assume there is an applicability mismatch at index $i$.
    Thus, there is an atom $A \in S$ s.t. $\pi_A(\sigma_i)$ (or $\pi_A(\sigma'_i)$) and no atom $B \in S$ with $\phi_A = \phi_B$ is applicable in $\sigma'_i$ (or $\sigma_i$).
    Assume w.l.o.g. that $\pi_A(\sigma_i)$ and not $\pi_A(\sigma'_i)$. The other case is symmetric.
    We distinguish two cases:
    \begin{compactenum}
        \item[\emph{Case 2.1:}] There is no atom $B \in \template$ with $\phi_A = \phi_B$ that is applicable in~$\sigma'_i$.
                        Then, atom~$A$ is strongly distinguishing.
        \item[\emph{Case 2.2:}] There is an atom $B \in \template$ with $\phi_A = \phi_B$ that is applicable in~$\sigma'_i$, but it is not included in $S$.
                        Then we further distinguish the following two cases:
                        \begin{compactenum}
                            \item[a)] $\phi_A(\sigma_i) = \phi_B(\sigma'_i)$. Then $(A, B)$ form a xor-distinguishing pair and exactly one of the two atoms is included in the contract.
                            \item[b)] $\phi_A(\sigma_i) \neq \phi_B(\sigma'_i)$. Then $A$ is strongly distinguishing. 
                        \end{compactenum}
    \end{compactenum}
    This concludes the proof.
\end{proof}

\subsection{Proof of \Cref{thm:template-distinguishability-date}}

\templatedistinguishabilitydate*

\begin{proof}
We show in the following Lemmas~\ref{lem:strongimpliesweak}, \ref{lem:weakmayimplystrong}, and \ref{lem:noxor} that if no two distinct atoms share the same leakage function, then 
\begin{enumerate}
    \item the notions of strongly-distinguishing atoms and distinguishing atoms coincide (Lemmas~\ref{lem:strongimpliesweak}+\ref{lem:weakmayimplystrong}), and
    \item there are no xor-distinguishing pairs (Lemma~\ref{lem:noxor}).
\end{enumerate}
Then \Cref{thm:template-distinguishability} reduces to \Cref{thm:template-distinguishability-date}.
\end{proof}

The proofs of Lemmas~\ref{lem:strongimpliesweak}--\ref{lem:noxor} are given below.

\begin{restatable}{lemma}{strongimpliesweak}\label{lem:strongimpliesweak}
    Every {strongly-distinguishing} atom is also {distinguishing}, but a distinguishing atom is not necessarily strongly distinguishing.
\end{restatable}

\begin{proof}
    For the second part of the lemma, consider again Example~\ref{example:template}.
    Both atoms in the example are distinguishing, but they are not strongly distinguishing.
    \par
    For the first part of the lemma, note that by definition atom $A$ is strongly distinguishing if there exists an index $i$, such that
     \begin{eqnarray*}
        \left(\pi_A(\sigma_i) \wedge \neg\exists {B \in \template}: \pi_B(\sigma'_i) \wedge \phi_A(\sigma_i) = \phi_B(\sigma'_i)\right)\\
        \vee \left(\pi_A(\sigma'_i) \wedge \neg\exists {B \in \template}: \pi_B(\sigma_i) \wedge \phi_A(\sigma'_i) = \phi_B(\sigma_i)\right).
     \end{eqnarray*}
    Applying De Morgan's law this is equivalent to 
     \begin{eqnarray*}
        \left(\pi_A(\sigma_i) \wedge \forall {B \in \template}: \neg\pi_B(\sigma'_i) \vee \phi_A(\sigma_i) \neq \phi_B(\sigma'_i)\right)\\
        \vee \left(\pi_A(\sigma'_i) \wedge \forall {B \in \template}: \neg\pi_B(\sigma_i) \vee \phi_A(\sigma'_i) \neq \phi_B(\sigma_i)\right).
     \end{eqnarray*}
    By the fact that $(\forall_{B \in \template} \psi(B)) \implies \psi(A)$ for any atom $A \in \template$ and formula $\psi$ this implies
     \begin{eqnarray*}
        \left(\pi_A(\sigma_i) \wedge (\neg\pi_A(\sigma'_i) \vee \phi_A(\sigma_i) \neq \phi_A(\sigma'_i))\right)\\
        \vee \left(\pi_A(\sigma'_i) \wedge (\neg\pi_A(\sigma_i) \vee \phi_A(\sigma'_i) \neq \phi_A(\sigma_i))\right),
     \end{eqnarray*}
    which is equivalent to the definition of atom $A$ being distinguishing:
     \begin{eqnarray*}
       \pi_A(\sigma_i) \oplus \pi_A(\sigma'_i) \vee (\pi_A(\sigma_i) \wedge \pi_A(\sigma'_i) \wedge \phi_A(\sigma_i) \neq \phi_A(\sigma'_i)).
     \end{eqnarray*}
\end{proof}

\begin{restatable}{lemma}{weakmayimplystrong}\label{lem:weakmayimplystrong}
    If no two distinct atoms share the same leakage function, 
    then {distinguishing} atoms are also {strongly distinguishing}.
\end{restatable}

\begin{proof}
    Recall the requirement, stated in \Cref{sec:contracts}, that the images of different leakage functions in a given contract template must be disjoint.
    If no two atoms share the same leakage function this in particular implies $\forall A,B \in \template, A \neq B: \phi_A(\sigma_i) \neq \phi_B(\sigma'_i)$ 
    and so we have $\phi_A(\sigma_i) \neq \phi_A(\sigma'_i) \Leftrightarrow \forall B \in \template: \phi_A(\sigma_i) \neq \phi_B(\sigma'_i) \Leftrightarrow \forall B \in \template: \phi_A(\sigma'_i) \neq \phi_B(\sigma_i)$.

    Plugging the above equivalences into the definition of $A$ being distinguishing we get
    \begin{multline*}
       \pi_A(\sigma_i) \oplus \pi_A(\sigma'_i) \vee (\pi_A(\sigma_i) \wedge \pi_A(\sigma'_i) \wedge \\ \forall B \in \template: \phi_A(\sigma_i) \neq \phi_B(\sigma'_i) \wedge \forall B \in \template: \phi_A(\sigma'_i) \neq \phi_B(\sigma_i))
    \end{multline*}   
    which implies that $A$ is strongly distinguishing.
\end{proof}

\begin{restatable}{lemma}{noxor}\label{lem:noxor}
    If no two distinct atoms share the same leakage function, then there are no xor-distinguishing atom pairs.
\end{restatable}

\begin{proof}
    This follows immediately from the requirement, stated in \Cref{sec:contracts}, that the images of different leakage functions in a given contract template must be disjoint.
\end{proof}

\subsection{Proof of \Cref{thm:ilp-synthesis}}

\ilpsynthesis*

\begin{proof}
Consider a solution to the ILP and the corresponding set of atoms $S \subseteq \template$ associated to the variables $s_A$ that are set to true in the solution.
Given an attacker-distinguishable test case~$t$, Constraint \ref{eq:coverage} ensures that a strongly-distinguishable atom or a xor-distinguishing atom pair are part of $\contract_S$.
From this and \Cref{thm:template-distinguishability}, $t$ is contract-distinguishable w.r.t. $S$.
Furthermore, the objective function $\min \sum_{t \in T_{nd}}{\textit{fp}_t}$ ensures that $S$ distinguishes as few attacker-indistinguishable test cases as possible.
\end{proof}

\subsection{Proof of \Cref{prop:bounded-verification:correctness-1}}
\boundedverificationcorrectnessa*

 \begin{proof}
    A counterexample $\mathit{cex}$ to $\phi^{C,\psi}_{\mathit{ctrsat}}$ represents an execution of the stuttering product circuit where (a) both sides have the same initial microarchitectural state $\mu_0$ and different architectural states $\sigma, \sigma'$, (b) contract equivalence holds for cycles $0$ to $k'$ whenever the retirement predicate $\psi$ holds on both sides (where $k'$ is the minimum between $k$ and the cycle when $i$ reaches $0$), and (c) $j$ is the smallest cycle  $0 \leq j \leq b$ where the two executions are attacker distinguishable.
    Hence, we can construct the test case $T = (\sigma,\sigma')$.
    Let $n \in \Nat$ be the number of times when $\psi \wedge \psi'$ holds in $\mathit{cex}$ up to cycle $k'$ included, \ie the number of instructions retired by both executions up to $k'$.
    From (b) and the correctness of the stuttering construction~\cite{Wang23}, we get that the prefixes of length $n$ of $C(\implEvalFilter{\phi}((\sigma,\mu_0)))$ and $C(\implEvalFilter{\phi}((\sigma',\mu_0)))$ are the same.
    Similarly, from (c) and the stuttering circuit's construction~\cite{Wang23}, we get that $\atk(\implEval((\sigma,\mu_0))) \neq \atk(\implEval((\sigma',\mu_0)))$.
    Thus, $T$ is attacker distinguishable even though the first $n$ instructions produce the same contract observations w.r.t. $C$.
    Now, we just need to show that $n \geq \mathit{min}(\lfloor \frac{k}{\mathbb{K}}\rfloor,  i)$.
    There are two cases depending on $k'$, \ie the cycle at which \tool{} stops.
    In both cases, w.l.o.g. we assume that the programs in the counterexample do not terminate before $k'$ cycles. %
    If $k' = k$, then the counterexample $\mathit{cex}$ executed at least $\lfloor \frac{k}{\mathbb{K}}\rfloor$ instructions (given our assumption on the maximum retirement time).
    If $k' \neq k$, then the counterexample $\mathit{cex}$ retired at least $i$ instructions (to decrease the instruction counter to $0$). 
    \end{proof}

\subsection{Proof of \Cref{prop:bounded-verification:correctness-2}}

\boundedverificationcorrectnessb*

\begin{proof}
Assume, for contradiction's case, that this is not the case.
That is, there exists a $C$-equivalent test case $T$ consisting of at most $\mathit{min}(\lfloor \frac{k}{\mathbb{K}}\rfloor,  i)$ instructions that is attacker distinguishable in the first $b$ cycles even though the BMC cannot disprove $\phi^{C,\psi}_{\mathit{ctrsat}}$.
From the assumption on the program's size and on retirement time, we know that the programs in $T$ can be fully executed in at most $\mathbb{K}\times \mathit{min}(\lfloor \frac{k}{\mathbb{K}}\rfloor,  i)$ cycles.
From this and the construction of the stuttering circuit~\cite{Wang23}, we know that $T$ maps to an execution of the stuttering circuit that would violate $\phi^{C,\psi}_{\mathit{ctrsat}}$, leading to a contradiction.
\end{proof}

\subsection{Proof of \Cref{prop:phase3:correctness}}

    \overallcorrectness*

    \begin{proof}
        Any contract that passes Phase 3 has been verified by \textsc{LeaVe} and, therefore, it is sound according to~\cite[Theorem 2]{Wang23}.
        Furthermore, the contract distinguishes as few attacker-distinguishable test cases as possible following \Cref{thm:ilp-synthesis}.
    \end{proof}
}

\section{Additional tables}\label{appendix:data}

\Cref{tab:contract:leaks} summarizes the leaks captured by the synthesized contracts in the experiment from \textbf{RQ1}.
\newcolumntype{C}[1]{>{\centering\arraybackslash}m{#1}}

\newcommand{\cmark}{\ding{51}}%
\newcommand{\xmark}{\ding{55}}%
\newcommand{\ctrp}{$\circ$\xspace}
\newcommand{\ctrf}{\cmark\xspace}
\newcommand{\ctrn}{\xmark\xspace}
\renewcommand{\ctrp}{$\exists$\xspace}
\renewcommand{\ctrf}{$\forall$\xspace}
\renewcommand{\ctrn}{$\nexists$\xspace}
\def\downcirc{\raisebox{1pt}{\scalebox{0.6}{\rotatebox[origin=c]{90}{\LEFTcircle}}}}
\makeatletter
\DeclareRobustCommand{\circbullet}{\mathbin{\vphantom{\circ}\text{\circbullet@}}}
\newcommand{\circbullet@}{%
  \check@mathfonts
  \m@th\ooalign{%
    \clipbox{0 0 0 {\dimexpr\height-\fontdimen22\textfont2}}{$\bullet$}\cr
    $\circ$\cr
  }%
}
\renewcommand{\ctrp}{{\Large$\circbullet$}\xspace}
\renewcommand{\ctrf}{{\Large$\bullet$}\xspace}
\renewcommand{\ctrn}{{\Large$\circ$}\xspace}
\newcommand{\ctra}{-}

\begin{table*}[h]
    \caption{The table shows the leaks captured by the synthesized contract.}\label{tab:contract:leaks}%
    \centering
    \begin{threeparttable}
        \small
    \begin{tabular}{lC{1cm}C{1cm}C{1cm}C{1cm}C{1cm}C{1cm}C{1cm}}
        \toprule
    \textit{Processor} &
    \rot{\begin{tabular}[c]{@{}l@{}}Branch \& Jump instr.:\\ \texttt{BRANCH_TAKEN}\tnote{a}\end{tabular}} & 
    \rot{\begin{tabular}[c]{@{}l@{}}\texttt{div}, \texttt{divu}, \texttt{rem}, \texttt{remu}:\\ \texttt{REG\_RS2\_ZERO} \end{tabular}} & 
    \rot{\begin{tabular}[c]{@{}l@{}}Load \& Store instr.:\\ distinguishable\end{tabular}} & 
    \rot{\begin{tabular}[c]{@{}l@{}}Load \& Store instr.:\\ alignment\end{tabular}} & 
    \rot{\begin{tabular}[c]{@{}l@{}}\texttt{mul}, \texttt{mulh}, \texttt{mulhu}, \texttt{mulhsu}:\\ distinguishable\end{tabular}} &
    \rot{\begin{tabular}[c]{@{}l@{}}Load instr.:\\ \texttt{MEM_ADDR}\end{tabular}} &
    \rot{\begin{tabular}[c]{@{}l@{}}\texttt{mul}:\\ \texttt{REG\_RS2\_LOG2}\end{tabular}} \\
    \midrule
    \DarkRiscvTwoStages{}   & \ctrf &   &   &                                                     &   &   &   \\
    \DarkRiscvThreeStages{} & \ctrf &   &   &                                                     &   &   &   \\
    \SodorTwoStages         & \ctrf &   &   &                                                     &   &   &   \\
    \IbexSmall{}            & \ctrf & \ctrf & \ctrf & \ctrf                                                   & \ctrf &   &   \\
    \IbexCache{}            & \ctrf & \ctrf & \ctrf & \ctrp\tnote{b} & \ctrf & \ctrf &   \\
    \IbexMul{}              & \ctrf & \ctrf & \ctrf & \ctrf                                                   & \ctrf &   & \ctrf \\
    \bottomrule
    \end{tabular}
    \begin{tablenotes}
        \footnotesize
        \item[a] For jump instructions, there are other atoms that yield the same precision (w.r.t. to our evaluation set). Therefore, in some runs (depending on the initial test cases) the contract may expose slightly different atoms.
        \item[b] Store instructions only, load instructions leak \texttt{MEM_ADDR} and the type of the load.
    \end{tablenotes}
\end{threeparttable}
\end{table*}

\section{Property encoding in Verilog } 
\label{appendix:property}

\Cref{fig:verilog-encoding} depicts the Verilog property encoding bounded contract satisfaction checked by \tool{}.
The property is expressed on top of the product circuit constructed as outlined in \Cref{sec:synthesis:bounded-verification}.

\begin{figure*}
    \begin{lstlisting}[style={verilog-style}]
    wire atk_equiv = ... // expression encoding that the attacker observations are the same on the current cycle
    wire ctr_equiv = ... // expression encoding that if the retirement predicate holds in the current cycle
    // in the two executions, then the contract observations must be the same

    // property checked
    reg [31:0] counter = b;
    reg [31:0] instr_counter = i;
    reg state_atk_equiv = 1 ;
    always @ (posedge clock) begin
        if (counter > 0) begin
            counter <= counter - 1;
            state_atk_equiv <= state_atk_equiv && atk_equiv;
        end
        // There is an attacker conflict and a new instruction is retired
        if (state_atk_equiv == 0 && Retire_left && Retire_right) begin
            instr_counter <= instr_counter - 1;
        end
    end
    // initialize microarchitectural states
    assume property (init_state_equiv); 
    // assume contract equivalence until the instruction counter equals to 0
    assume property (instr_counter == 0 || ctr_equiv); 
    // configuration of the processor
    assume property (state_invariant); 
    // check the attacker equivalence from cycle 0 to b
    assert property (instr_counter != 0 || state_atk_equiv); 
    \end{lstlisting}
    \vspace{-10pt}
    \caption{Encoding of the bounded contract satisfaction property in Verilog}\label{fig:verilog-encoding}
\end{figure*}

}
\end{document}